\documentclass{article}
\usepackage[margin=1.25in]{geometry}
\title{Control of fluctuations and heavy tails for heat variation in the two-time measurement framework}
\author{Tristan Benoist\textsuperscript{1}, Annalisa Panati\textsuperscript{2}, Renaud Raqu\'epas\textsuperscript{3,4}}

\usepackage[colorlinks = true, linkcolor = teal, urlcolor  = teal, citecolor = violet]{hyperref}

\usepackage{amsmath,amsfonts,amsthm, amssymb}
\usepackage{bbm,bm}
\usepackage{braket}
\usepackage[usenames,dvipsnames,svgnames,table]{xcolor}
\usepackage{enumitem}

\theoremstyle{definition}
	\newtheorem{definition}{Definition}[section]
	\newtheorem{remark}[definition]{Remark}
	\newtheorem{example}[definition]{Example}
	\newtheorem*{example*}{Example}

\theoremstyle{plain}
	
	\newtheorem{lemma}[definition]{Lemma}
	\newtheorem*{lemma*}{Lemma}
	\newtheorem{proposition}[definition]{Proposition}
	\newtheorem{theorem}[definition]{Theorem}
	\newtheorem*{theorem*}{Theorem}
	
	\newtheorem{corollary}[definition]{Corollary}

\makeatletter
\newcommand{\specialcell}[1]{\ifmeasuring@#1\else\omit$\displaystyle#1$\ignorespaces\fi}
\makeatother

\newcommand{\numset}[1]{\mathbf{#1}}
	\newcommand{\cc}{\numset{C}}
	\newcommand{\rr}{\numset{R}}

	\newcommand{\nn}{\numset{N}}

\newcommand{\one}{\bm{1}}
\newcommand{\e}{\mathrm{e}}
	\newcommand{\Exp}[1]{\mathrm{e}^{#1}}
\renewcommand{\i}{\mathrm{i}}
\renewcommand{\Im}{\operatorname{Im}}

\makeatletter
\providecommand*{\diff}
	{\@ifnextchar^{\DIfF}{\DIfF^{}}}
\def\DIfF^#1{
	\mathop{\mathrm{\mathstrut d}}
		\nolimits^{#1}\gobblespace}
\def\gobblespace{
	\futurelet\diffarg\opspace}
\def\opspace{
	\let\DiffSpace\!
	\ifx\diffarg(
		\let\DiffSpace\relax
	\else
		\ifx\diffarg[
			\let\DiffSpace\relax
		\else
			\ifx\diffarg\{
				\let\DiffSpace\relax
			\fi\fi\fi\DiffSpace}

\providecommand*{\od}[3][]{
	\frac{\diff^{#1}#2}{\diff #3^{#1}}}

\renewcommand{\d}{\diff}

	\renewcommand{\P}{\mathbb{P}}
	\newcommand{\cht}{\mathcal{E}_t}
	\newcommand{\ee}{\mathbb{E}}

	\renewcommand{\sp}{\operatorname{sp}}
	
	\DeclareMathOperator{\supp}{supp}
	\DeclareMathOperator{\Dom}{Dom}

	\DeclareMathOperator{\tr}{tr}

	\newcommand{\cH}{\mathcal{H}}
	\renewcommand{\H}{\mathcal{H}}
	
	\newcommand{\h}{\mathfrak{h}}
	\newcommand{\ch}{\mathfrak{h}}
	\newcommand{\B}{\mathcal{B}}
			\newcommand{\cB}{\mathcal{B}}

	\newcommand{\cO}{\mathcal{O}}

	\newcommand{\Ga}[1][]{\Gamma^{-}_{#1}}
	\newcommand{\Gs}[1][]{\Gamma^{+}_{#1}}
	\newcommand{\G}[1][]{\Gamma_{#1}}
	\newcommand{\dG}[1][]{\diff\Gamma_{#1}}

   \newcommand{\norme}[1]{\|#1\|}

	\newcommand{\f}{\psi_f}
   \def\ones{\psi_{\textnormal{o}}}
   \newcommand{\eno}{\epsilon_{\textnormal{o}}}

	\newcommand{\cL}{\mathcal{L}}

	\newcommand{\wlim}{\operatorname*{w-lim}}
	
	\newcommand{\srlim}{\operatorname*{s-r-lim}}
	\newcommand{\ad}{\operatorname{ad}}

\begin{document}

\maketitle

\begin{center}
\small
\begin{tabular}{c c c}
   1. Universit\'e de Toulouse UPS & & 2. Aix Marseille Univ, Universit\'e de Toulon \\
   CNRS, Institut de math{\'e}matiques de Toulouse & &  CNRS, CPT \\
   F-31\,062 Toulouse Cedex 9, France & & Marseille, France \\
   && \\
   3. McGill University & & 4. Univ. Grenoble Alpes \\
   Department of Mathematics and Statistics & & CNRS, Institut Fourier \\
   1005--805 rue Sherbrooke~Ouest & & F-38\,000 Grenoble, France \\
    Montr{\'e}al (Qu{\'e}bec) ~H3A 0B9 & & \\
\end{tabular}
\end{center}

\begin{abstract}
   We study heat fluctuations in the two-time measurement framework. For bounded perturbations, we give sufficient ultraviolet regularity conditions on the perturbation for the moments of the heat variation to be uniformly bounded in time, and for the Fourier transform of the heat variation distribution to be analytic and uniformly bounded in time in a complex neighborhood of~$0$.

   On a set of canonical examples, with bounded and unbounded perturbations, we show that our ultraviolet conditions are essentially necessary. If the form factor of the perturbation does not meet our assumptions, the heat variation distribution exhibits heavy tails. The tails can be as heavy as preventing the existence of a fourth moment of the heat variation.
\end{abstract}

\section{Introduction}
Experimental advances in the control of mesoscopic systems are fueling the interest for the study of fluctuations of thermodynamic quantities. Indeed, while the number of degrees of freedom in mesoscopic systems is large enough to justify a thermodynamical analysis, it is still small enough that fluctuations of thermodynamic quantities about their mean are relevant. In this paper we are interested in the fluctuations of the quantities entering the first law of thermodynamics for isolated locally perturbed systems. We are more precisely focusing on heat fluctuations in quantum systems.

\medskip
In isolated systems\,---\,both classical and quantum\,---\,energy conservation holds almost surely, which is typically written as $\Delta U = 0$ in the context of thermodynamics. In classical systems, the first law of thermodynamics can then be expressed more precisely as the almost sure equality between the work~$W$ and the heat~$\Delta Q$ random variables: $0 =  \Delta U = \Delta Q - W$.
For quantum systems, the translation of the definitions of such thermodynamic random variables is not straightforward. In particular, depending on the definition, different classical thermodynamic relations may fail to hold beyond the level of averages.

Typically, a naive quantization of the classical definition, while preserving $\Delta Q=W$ in law, is hard to interpret physically and leads to a failure of the celebrated classical fluctuation relations~\cite{ECM,GC1,GC2,Jar,Crooks}; see \cite{TLH,JOPP}. To circumvent this issue, Kurchan~\cite{Kur} and Tasaki~\cite{Tas} proposed in~2000 a defintion of thermodynamic random variables as differences between two outcomes of energy measurement. We will call this approach the two-time measurement (TTM) framework.\footnote{It is also, sometimes, called Full Counting Statistics.} This definition has a clear physical interpretation and fluctuation relations can easily be derived in this framework. It thus has grown into a lively research subject; see the reviews \cite{EHM,CHT} and \cite{JOPP} for a more mathematically oriented approach.
Although successful with regard to fluctuation relations, and while preserving equality of the mean work and heat, the TTM definitions give distinct laws to~$\Delta Q$ and~$W$.
Therefore, statistical fluctuations of heat and work have to be studied independently.

The above picture motivates the study of heat fluctuations alone that we provide in this contribution. More precisely, we are interested in controlling the tails of the distribution of heat variations. We particularly aim at highlighting the difference between the classical framework and the quantum TTM framework. The work fluctuations are of minor interest to us since, for a bounded perturbation, work is almost surely uniformly bounded in time\,---\,in classical systems and for both the naive quantization and TTM definitions in quantum systems.
\medskip

For an isolated classical system\,---\,whose dynamics is governed by a time-independent perturbed Hamiltonian $H_\lambda= H_0+\lambda V$ generating a flow~$(\Xi_\lambda^t)_{t \in \rr}$\,---\,the heat variation $\Delta Q :=H_0\circ\Xi_\lambda^t-H_0$ is equal to the work $ W=-(\lambda V\circ\Xi_\lambda^t-\lambda V)$ by the invariance of $H_\lambda$ under the flow $(\Xi_\lambda^t)_{t\in\rr}$.
Following standard classical thermodynamics, the work is defined as an integral along a trajectory, that is $W:=-\int_0^t \lambda \frac{\d}{\d t} V(x_t)  \d t$. The equivalence of those two definitions of work follows from the fundamental theorem of calculus.

Clearly, if the perturbation~$V$ is bounded, the work~$W$ is surely bounded uniformly in time. So is the heat variation~$\Delta Q$, by their sure equality. Therefore,
boundedness of the perturbation~$V$ implies that the distribution~$\P_t$ of heat variation is compactly supported. If~$V$ is unbounded, the tails of~$\P_t$ are expected to be controlled by the strength parameter~$\lambda$.
We illustrate this with two examples in Section~\ref{sec:class-desc}.
\medskip

For a quantum system, the heat variation is defined in the TTM framework according to the following \emph{Gedankenexperiment}. A first measurement of~$H_0$, the unperturbed Hamiltonian, is performed at an initial time, yielding a result $E\in\operatorname{spec}H_0$; the system then evolves according to~$H_\lambda$, the perturbed Hamiltonian, for a time~$t$; and~$H_0$ is once again measured, yielding a result $E'\in\operatorname{spec}H_0$. The resulting heat variation is defined as the difference~$\Delta Q :=E'-E$ between the two measurement outcomes.

For confined systems, where the unperturbed Hamiltonian~$H_0$ and the density matrix~$\rho$ representing the initial state are commuting matrices, the characteristic function of the measure~$\P_t$ assigning probabilities to the differences~$\Delta Q$ in the above thought experiment is
$$\cht(\alpha):=\tr(\e^{\i t H_\lambda}\e^{\i\alpha H_0}\e^{-\i t H_\lambda}\e^{-\i \alpha H_0}\rho).$$

A rewriting of this expression in terms of algebraic objects that survive the thermodynamic limit serves as the basis for the construction of the generalization of this measure to infinitely extended systems; see Section~\ref{sec:fcs-setup}.

Mirroring the classical equality $ W=-(\lambda V\circ\Xi_\lambda^t-\lambda V)$, we define the work~$W$ in the TTM framework
according to a similar \emph{Gedankenexperiment},
but with measurements of~$-\lambda V$ instead of~$H_0$.\footnote{In the expression of the characteristic function $\rho$ also needs to be substituted for its projection onto the matrices commuting with $V$.} This definition ensures that $\Delta Q=W$ in mean regardless of the initial state $\rho$.
This equality in mean however does not extend to an equality in law. The characteristic function of these two probability measures are, in general, not equal.
Hence, while~$V$ being bounded implies $|W|\leq 2|\lambda|\|V\|$ almost surely, it does not imply a control of the tails of~$\P_t$. Other arguments are needed.
\medskip
Certainly motivated by the classical picture, it is colloquially assumed that, in the TTM framework, the tails of~$\P_t$ are still mainly controlled by~$\lambda$; see \cite[\S III.B.3]{EHM}. In this contribution, we challenge this assumption. We  provide sufficient conditions, in the thermodynamic limit, for the control of the tails of~$\P_t$ that do not involve the value of~$\lambda$ as long as~$V$ is bounded. We furthermore test the necessity of these conditions in different canonical models of quantum statistical mechanics. We aim at  conveying that smallness of~$\lambda$ is not the relevant condition to impose in order to prevent large heat fluctuations. Instead, a crucial role is played by the ultraviolet (UV) regularity of~$V$.

\medskip
A sufficient condition for the exponential control of the tails of~$\P_t$ has been introduced in our partly co-authored work~\cite{BJPPP}. There, we considered quantum dynamical systems arising as the limit of a sequence of finite dimensional systems. In the present article, we work directly in the thermodynamic limit, via the operator algebraic formulation of quantum statistical mechanics. 
Let~$\delta$ be the generator of the unperturbed dynamics on the $C^*$-algebra $\cO$ of observables of the system (in finite dimension, $\delta({\,\cdot\,})=\i[H_0,{\,\cdot\,}]$). Assume that the initial state $\omega$ over~$\cO$ is invariant under the unperturbed dynamics, {i.e.} $\omega\circ\e^{t\delta}=\omega$ for all~$t\in\rr$. Let~$V$ be a self-adjoint element of~$\cO$ defining the perturbation. In the language of the present paper, the result of~\cite{BJPPP} can be summarized by the implication
 \begin{equation}
 \label{eq:BJPPP_implication}
V\in\Dom \e^{\i\frac12\gamma \delta}\cap \Dom \e^{-\i\frac12\gamma \delta} \implies \sup_{t\in\rr}\ee_t (\e^{\gamma|\Delta Q|}) < \infty
 \end{equation}
 with $\ee_t$ the expectation with respect to $\P_t$ and $\gamma\in\rr_+$. We provide a proof of this implication in our context in Section~\ref{sec:sufficient_cond}. In terms of control of the tails, Markov's inequality implies that if the left-hand side of~\eqref{eq:BJPPP_implication} holds for some $\gamma>0$, then there exists $C>0$ such that
 $$\sup_{t\in\rr}\P_t(|\Delta Q|>E)\leq C\e^{-\gamma E}$$
 for any $E>0$.

It follows trivially from~\eqref{eq:BJPPP_implication} that if the left-hand side holds, any moment $\ee_t(|\Delta Q|^m)$ is uniformly bounded in time. In Theorem~\ref{thm:diff}, we show that such a uniform bound for the moments holds under weaker conditions:
\begin{equation}
\label{eq:implication_moments}
V\in\Dom \delta^n\implies \sup_{t\in\rr}\ee_t(\Delta Q^{2n+2})<\infty
\end{equation}
for any $n\in\nn$. In terms of control of the tails, Markov's inequality implies that if $V \in \Dom \delta^n$ for some $n\in\nn$, then there exists $C>0$ such that
$$\sup_{t\in\rr}\P_t(|\Delta Q|>E)\leq C \, E^{-2n-2}$$
for any $E>0$.

Remark that both the left-hand side of \eqref{eq:BJPPP_implication} and \eqref{eq:implication_moments} do not depend on the initial state. To obtain these implications, we merely assume the initial state is invariant under the unperturbed dynamics. In particular, the implications hold even if the system is not initially at equilibrium. Furthermore, in each statement, the left-hand side does not depend on $\lambda$; only the value of the supremum in the corresponding right-hand side depends on it.

The implications of \eqref{eq:BJPPP_implication} on the large deviations of~$\P_t$ were discussed in \cite{BJPPP}. In \cite{BPP}, two of us and Y. Pautrat discuss the implications of an assumption akin to the left-hand side of \eqref{eq:BJPPP_implication} on the large deviations of heat currents in autonomous out of equilibrium open systems. In particular, we prove a translation symmetry of the cumulant generating function, first proposed in \cite{andrieux2009fluctuation}. We relate it to the large deviations of the conservation of heat currents and derive a part of the fluctuation-dissipation theorem under a time-reversal invariance assumption.

\medskip
Both implications \eqref{eq:BJPPP_implication} and \eqref{eq:implication_moments} establish sufficient conditions for the control of the tails of $\P_t$. In the third part of the paper, starting in Section \ref{sec:models}, we establish an essential necessity of these conditions for different usual models of quantum statistical mechanics, with $V$ bounded or unbounded. Our first model corresponds to a fermionic impurity interacting with a quasi-free Fermi gas at equilibrium. The second model is the bosonic counterpart of the first one and is similar to the interaction of an harmonic oscillator with a quasi-free Bose gas. The third model is a van Hove Hamiltonian which corresponds to a quasi-free Bose gas at equilibrium interacting with prescribed classical charges.
The classical counterparts of these two last models of bosons are discussed in Sections~\ref{sec:model-classic_quad} and~\ref{sec:class-linear} respectively.

In the first model, a quasi-free Fermi gas, $V$ is bounded and, in Theorem~\ref{thm:fermions-equiv-moments}, we show the equivalence
\begin{equation}
\label{eq:equiv_moments_fermions_intro}
\begin{split}
   V\in\Dom \delta^n
      &\iff \exists\ t_1<t_2\mbox{ : }\int_{t_1}^{t_2} \ee_t(\Delta Q^{2n+2}) \d t<\infty \\
      &\iff \sup_{t\in\rr}\ee_t (\Delta Q^{2n+2})<\infty.
\end{split}
\end{equation}
In this model, the assumption $V\in\Dom \delta^n$ is equivalent to $\int_{\rr_+} e^{2n}|f(e)|^2\d e<\infty$ with $e\mapsto |f(e)|^2$ the form factor of the perturbation~$V$, where $e\in\rr_+$ is the energy.
Similarly, the assumption $V\in\Dom \e^{\i\frac12\gamma \delta}\cap \Dom \e^{-\i\frac12\gamma \delta}$ is equivalent to $\int_{\rr_+} \e^{\gamma e}|f(e)|^2\d e<\infty$.
Hence the left-hand side of \eqref{eq:BJPPP_implication} and \eqref{eq:implication_moments} indeed correspond to UV regularity conditions.
It is with similar models in mind that we choose to refer to the left-hand sides of \eqref{eq:BJPPP_implication} and \eqref{eq:implication_moments} as UV regularity conditions.
Note however that for quantum spin systems, these conditions are typically the consequence of some locality assumptions for $V$ and $H_0$; see \cite[\S 6.2]{BR2} and \cite[\S 4.1]{BPP}.

In the second and third models, $V$ is unbounded and $e\mapsto |f(e)|^2$ is again its form factor. For both these models, we prove the equivalence
\begin{equation}
\label{eq:equiv_moments_bosons_intro}
\begin{split}
   \int_{\rr_+}e^{2n}|f(e)|^2\d e<\infty
      &\iff \exists\ t_1<t_2\mbox{ : }\int_{t_1}^{t_2} \ee_t(\Delta Q^{2n+2}) \d t<\infty\\
      &\iff \sup_{t\in\rr}\ee_t (\Delta Q^{2n+2})<\infty.
\end{split}
\end{equation}
For the first bosonic model, we prove an implication similar to \eqref{eq:BJPPP_implication}, namely
\begin{equation}\label{eq:implication_fourier_bosons_intro}
\int_{\rr_+}\e^{\gamma e}|f(e)|^2\d e<\infty\implies \exists\ \gamma'\in(0,\gamma)\mbox{ : } \sup_{t\in\rr}\ee_t(\e^{\gamma' |\Delta Q|})<\infty.
\end{equation}
For the second bosonic model, we prove the equivalence
\begin{equation}\label{eq:equiv_exp_bosons_intro}
\begin{split}
   \int_{\rr_+}\e^{\gamma e}|f(e)|^2\d e<\infty
      &\iff \exists\ t_1<t_2\mbox{ : }\int_{t_1}^{t_2}\ee_t(\e^{\gamma|\Delta Q|})\d t<\infty \\
      &\iff \sup_{t\in\rr}\ee_t(\e^{\gamma |\Delta Q|})<\infty.
\end{split}
\end{equation}
For this last model, we moreover prove that~$\P_t$ is the law of an inhomogeneous Poisson process for every~$t$, while its classical counterpart is the law of a Gaussian random variable.

\medskip
The existence of these equivalences highlights the contrast between classical statistical mechanics and the quantum TTM framework. For a quantum system, if one aims at controlling the tails of the heat variation, before discussing the value of~$\lambda$, one has to assume that the interaction has sufficient UV regularity. Typically, if one uses a cutoff $f(e)=0$ for~$e$ large enough, completely preventing the contribution of energy scales at which the validity of the model is no longer guaranteed, then the left-hand side of \eqref{eq:BJPPP_implication} holds for any $\gamma>0$. This discussion is not necessary for classical systems.

Heuristically speaking, the physical picture can be understood in terms of the Fermi golden rule. According to this rule, the transition rates between different energy levels $E$ and $E'$ is given roughly by $T(E',E):=|\langle E', V E\rangle|^2$ with $\{|E\rangle\}_E$ the energy eigenstates of~$H_0$.
The underlying physical intuition for the left-hand side of~\eqref{eq:BJPPP_implication} [resp. \eqref{eq:implication_moments}] is roughly a condition controlling $T(E',E)\e^{\gamma |E'-E|}$ [resp. $T(E',E)|E'-E|^{2n}$] for large values of $|E-E'|$.
Hence, these assumptions control the decay of the transition rates as $|E'-E|$ goes to~$\infty$.

Note that, while important for the control of the tails at finite $\lambda$, the UV regularisation assumptions become irrelevant in the limit $t\to\infty$ and then $\lambda\to0$, if one assumes return to equilibrium. Indeed, adapting \cite{JPPP}, it can be shown that in this limit $\P_t$ converges weakly to a Dirac measure in $0$.

\medskip
As pointed out in the conclusion of \cite{BFJP1}, though $\P_t$ is generally thought not directly accessible as it involves the projective measurement of non-local quantities, proposals have emerged allowing for a sampling of~$\P_t$ using an interaction with an auxiliary qubit~\cite{DCH, MDCP, CBK, RCP}. The proposal of~\cite{DCH} involves only a local interaction between the qubit and the system and therefore appears more appropriate for thermodynamic systems. Hence,~$\P_t$ may be experimentally accessible indirectly and a phenomenon akin to heavy tails may be observed.

\paragraph{Structure.}
The paper is organised  as follows. In Section~\ref{sec:class-desc}, we discuss shortly classical systems with emphasis on two models with~$V$ unbounded whose quantum counterpart will be studied in Section~\ref{sec:models}.
In Section~\ref{sec:sufficient_cond}, we introduce the $C^*$-algebraic formalism of quantum statistical mechanics and prove~\eqref{eq:BJPPP_implication} and~\eqref{eq:implication_moments}.
In Section~\ref{sec:models}, we introduce the three models for which we test the necessity of the UV regularity conditions.
The proof of the implication~\eqref{eq:implication_fourier_bosons_intro} and the equivalences~\eqref{eq:equiv_moments_fermions_intro}, \eqref{eq:equiv_moments_bosons_intro} and~\eqref{eq:equiv_exp_bosons_intro} for these models are postponed to Section~\ref{sec:models-proofs}.
Appendix~\ref{app:self-adjoint} contains technical results on the self-adjointness of some operators involved in the definition of~$\P_t$ for unbounded~$V$.
Finally, in Appendix~\ref{app:TL}, we show that the measure~$\P_t$ defined in our models through the algebraic formalism indeed emerges as the thermodynamic limit of measures describing two-time measurement protocols for some finite dimensional systems.

\paragraph{Acknowledgements.} We would like to thank Vojkan Jak\v{s}i\'{c} for informative discussions and suggestions about this project, and Claude-Alain Pillet and Yan Pautrat for useful comments about our work.
The research of T.B. has been supported by ANR-11-LABX-0040-CIMI within the program ANR-11-IDEX-0002-02 and ANR project StoQ (ANR-14-CE25-0003-01).
The research of A.P. was partially supported by ANR project SQFT (ANR-12-JS01-0008-01) and ANR grant  NONSTOPS (ANR-17-CE40-0006-01, ANR17-CE40-0006-02,
ANR-17-CE40-0006-03).
The research of {R.R.} was partially supported by NSERC and FRQNT. {R.R.} would like to thank the {Laboratoire de physique th\' eorique at Universit\'{e} Paul Sabatier, where part of this research was conducted, for its support and hospitality.

\section{A short classical detour}
	\label{sec:model-gauss}\label{sec:classic}
		\label{sec:class-desc}
\newcommand{\freq}{\nu}
To give some perspective on our results we first make a short detour through classical statistical mechanics. We refer the reader to \cite[\S 3]{thirringbook} for a detailed definition of classical Hamiltonian dynamical systems and to~\cite{ChMa} for an overview of the possible problems arising when considering an infinite-dimensional phase space.

Let $(K,\varpi)$ be a connected smooth symplectic manifold. The manifold~$K$ is endowed with the meaning of a classical system phase space. Since we are concerned with a thermodynamic setting, the manifold~$K$ is not assumed to be finite dimensional.
The unperturbed Hamiltonian of the system is a continuously differentiable function $H_0:K\to\rr$. The perturbation is a second continuously differentiable function $V:K\to\rr$, giving rise to the perturbed Hamiltonian~$H:=H_0+V$.

Under minimal technical assumptions on~$H$, it defines a continuous Hamiltonian flow~$(\Xi^t)_{t \in \rr}$ defined on~$K$, and the evolution of any sufficiently regular function~$g$ satisfies
$$
	\od{}{t}(g \circ \Xi^t) = \{g , H\} \circ \Xi^t,
$$
where~$\{{\,\cdot\,},{\,\cdot\,}\}$ is the Poisson bracket induced by~$\varpi$; see the general conservation theorem in~\cite{ChMa}.

We define the heat variation between times~$0$ and~$t$ as the difference between the unperturbed energy at times~$t$ and~$0$:
$
	\Delta Q:=H_0\circ\Xi^t - H_0
$.
We give some arguments for this choice in Remark~\ref{rk:heat_def}. Since $H\circ\Xi^t=H$ and $H=H_0+V$, equivalently,
\begin{equation}\label{eq:def_DQ_classic}
\Delta Q:=V-V\circ\Xi^t.
\end{equation}
Note that under minimal technical assumptions\footnote{See the conservation theorems in~\cite{ChMa}.} on~$V$,
$$\Delta Q = -\int_0^t\frac{\d}{\d s}V\circ\Xi^s \d s=-\int_0^t \{V,H_0\}\circ\Xi^s \d s.$$
This equality is the expression of the first law of thermodynamics for a closed system, the right-hand side being the work.

This definition through the work and the first law is useful when one is concerned with thermodynamic systems. Indeed, in some standard examples, the phase space needs to be enlarged to $\bar K\supset K$, a measurable space in which~$K$ is dense (see for example~\cite{JPActa,LSJSP}). Then, the functions $H_0$ and $V$, and the flow $\Xi^t$ are typically extended by continuity to $\bar K$. Though~$H_0$ may be infinite at some points in~$\bar K$, we assume that the perturbation~$V$ remains finite everywhere on $\bar K$. The function~$\Delta Q$ on~$K$ is then extended by continuity to a function on~$\bar K$ using the extensions of~$V$ and~$\Xi^t$.

With these definitions in mind, we introduce the statistical aspect of the model: the initial state of the system is a probability measure~$\mu$ on the phase space~$\bar K$ and the heat variation~$\Delta Q$ becomes a random variable. We are interested in characterizing the law~$\P_t$ of~$\Delta Q$ with respect to~$\mu$.

\begin{remark}
	We denote $\ee_t$ the expectation with respect to~$\P_t$. The random variable~$\Delta Q$ is denoted without index $t$ to highlight that we do not study a stochastic process but only a family~$(\P_t)_{t \in \rr}$ of measures on~$\rr$.
\end{remark}

\begin{remark}
	Whenever the initial definition $\Delta Q =H_0 \circ \Xi^t - H_0$ makes sense,~$\P_t$ has characteristic function
	\begin{equation}\label{eq:two_time_classic}
		\ee_t(\Exp{\i \alpha \Delta Q}) = \mu(\Exp{\i \alpha (H_0 \circ \Xi^t - H_0)}) = \mu(\Exp{\i \alpha H_0} \circ \Xi^t\  \Exp{-\i \alpha H_0}).
	\end{equation}
	 This formula parallels an expression we will encounter for quantum systems (see Remark~\ref{rk:finite_dim}).
\end{remark}

The following proposition is a trivial consequence of the definition of~$\Delta Q$.
\begin{proposition}\label{prop:classical_bounded_V}
	If $C:=\sup_{x\in\bar K}|V(x)|<\infty$, then  for all~$t\in\rr$, $|\Delta Q|<2C$ surely.
\end{proposition}
As soon as the perturbation is bounded, the heat fluctuations are surely uniformly controlled in time. We will see in Section \ref{sec:model-fermions} that this implication is false in the context of quantum two-time measurements.

\begin{remark}\label{rk:heat_def}
The separation between work and heat involves a certain amount of arbitrariness. More precisely, the choice of $H_0$ and $V$ such that $H=H_0+V$ is partly arbitrary. There is no reason another couple of functions would not be preferred. For example, if $W:K\to \rr$ is a continuously differentiable function that remains finite on $\bar K$, then $H_0+W$ and $V-W$ could also be a suitable couple of functions. The choice of $H_0$ in the definition of heat must therefore be motivated by other arguments that are a priori model dependent.

Nevertheless, since~$H$ is not time dependent, a quite interesting generic argument is given by the entropy balance equation. Assume, if it exists, that~$\mu$ is the Gibbs measure at inverse temperature~$\beta>0$ for the Hamiltonian~$H_0$. Then the choice of the splitting $H=H_0+V$ for the definition of~$\Delta Q$ is motivated by the $\mu$-almost sure validity of the entropy balance equation
\begin{equation} \label{eq:ent-bal-classic}
		\beta \Delta Q=-\log \frac{\d \mu_{-t}}{\d\mu},
\end{equation}
where $\mu_s := \mu \circ \Xi^{-s}$ is the time evolution of the measure $\mu$ induced by the flow~$(\Xi^{s})_{s\in\rr}$. Averaging, we recover the mean entropy balance equation
\begin{equation} \label{eq:ent-bal-classic-avg}
	\beta \ee_t(\Delta Q) = -\int \log \frac{\d \mu_{-t}}{\d\mu} \d\mu =: S(\mu|\mu_{-t})
\end{equation}
where $S(\mu|\nu)$ is the relative entropy of $\mu$ with respect to $\nu$.
A similar argument holds for composite systems with different parts at different inverse temperatures.

If another choice of splitting of $H$ is made, Equation~\eqref{eq:ent-bal-classic} does not hold exactly and some boundary terms may have non-trivial asymptotic contributions; see \cite{JPSHarmonic}.
\end{remark}

To the authors' knowledge, if $V$ is unbounded, the existence of the moments or the extension to an open set of $\cc$ of the Fourier transform of $\P_t$  can only be discussed on a model-by-model basis. In the next two subsections, we discuss two models with unbounded~$V$ that can be studied using only properties of Gaussian random variables.

\subsection{Harmonic systems}\label{sec:model-classic_quad}
Let~$Y$ be a real separable Hilbert space. Let~$B$ be a normal operator on~$Y$ with domain $\Dom B$ and trivial kernel. Let~$X$ be the completion of~$\Dom B$ with respect to the norm~$\|B{\,\cdot\,}\|_Y$. We also denote the extension of~$B$ to~$X$ by the letter~$B$. We interpret the vector space $K=Y\oplus X$ as the phase space of a
collection of harmonic oscillators. We endow~$K$ with the Hilbert space structure given by the inner product
$$
	\langle \pi\oplus \phi, \pi'\oplus \phi'\rangle:=\langle \pi,\pi'\rangle_Y + \langle B \phi,B \phi'\rangle_Y.
$$
Let~$\varpi$ be the symplectic bilinear form on~$K$ defined by
$
	\varpi(x,y)=\langle\mathcal{L}_0^{-1} x,y\rangle,
$
where
$$
	\mathcal{L}_0
		:=\begin{pmatrix}
				0 & -B^* B\\ \one & 0
			\end{pmatrix}.
$$
The operator $\mathcal{L}_0$ is skew-adjoint with domain
$\Dom(\mathcal{L}_0)=\{(\pi,\phi)\in K\ : \pi\in \Dom B, B\phi \in \Dom B^*\}.$
The space $(K,\varpi)$ is our symplectic (Hilbert) manifold.

Let the unperturbed Hamiltonian be defined on~$K$ by
\begin{align*}
	H_0 :
		x &\mapsto \tfrac{1}{2} \|x\|^2.
\end{align*}

Let $v$ be a trace class self-adjoint operator on $K$ such that $\|\cL_0 v\|<\infty$. The perturbation $V$ is then defined on~$K$ as
\begin{align*}
 V:
  	x&\mapsto \tfrac12\langle x,v x\rangle,
\end{align*}
and $H(x) = \frac12\langle x,(\one+v) x\rangle$.

Let $\mathcal{L}=\mathcal{L}_0(\one+v)$. Since $\mathcal{L}_0$ is skew-adjoint and  $\|\mathcal{L}_0 v\|<\infty$, perturbation theory implies that~$(\e^{t\mathcal{L}})_{t \in \rr}$ is a semigroup of bounded operators on~$K$. It is easy to check that the flow~$(\Xi^t)_{t\in\rr}$ associated to~$H$ is then given by
$\Xi^t (x) = \e^{t \mathcal{L}}x$
for any~$x\in K$.
Moreover, using the inequality $\|A^*TA\|_{\tr}\leq \|A\|^2\|T\|_{\tr}$ for $A\in\cB(K)$ and $T$ trace class, $v_t:=\e^{t\cL^*}v\e^{t\cL}$ remains trace class for any $t\in\rr$. Here, $\|{\,\cdot\,}\|_{\tr}$ denotes the trace norm.

\medskip
We turn to the construction of Gaussian states. First,~$K$ needs to be enlarged. Since it is separable by assumption, there exists an orthonormal basis~$(u_i)_{i\in\nn}$ of~$K$. Any such countable orthonormal basis (ONB) induces an isomorphism of $K$ into $\ell^2(\nn; \rr)$: $K\ni x\mapsto (\langle u_i,x\rangle)_{i\in\nn}$.
Let $(l_i)_{i\in\nn}$ be a sequence of strictly positive numbers such that $\sum_{i\in\nn} l_i=1$. Let $\langle{\,\cdot\,},{\,\cdot\,} \rangle_l$ be the inner product on~$K$ defined by
\begin{equation*}
	\braket{x,y}_l := \sum_{i \in \nn} l_i \braket{u_i,x} \braket{u_i,y}
\end{equation*}
and let~$\bar K$ be the completion of~$K$ with respect to the norm induced by $\langle{\,\cdot\,},{\,\cdot\,} \rangle_l$.

Consider a positive definite bounded linear operator~$D$ on~$K$.
As pointed out for example in~\cite{JPR}, by Kolmogorov's extension Theorem, there is a unique measure~$\mu_D$ on~$\bar K$ whose restriction to $\operatorname{linspan}\{u_i\}_{i \in I}$, with $I \subset \nn$ a finite set, is the centered Gaussian measure with inverse covariance matrix~$[\braket{u_i, D u_j}]_{i,j \in I}$.
The fact that~$\mu_D$ is indeed concentrated on $\bar K$ is guaranteed by the estimate~$\sum_{n \in I} l_n \braket{u_n, D u_n} \leq \|D\|$. By uniqueness of the measure, the inclusion $\supp \mu_D \subseteq \bar K$  does not depend on the choice of $(u_i)_{i \in\nn}$ and $(l_i)_{i\in \nn}$.

The measure $\mu_D$ then has characteristic function
\begin{equation*}
\begin{split}
\xi_D: 		y&\mapsto \e^{-\frac12 \langle y,Dy\rangle},
\end{split}
\end{equation*}
defined on $\bar K^*\subset K$, the dual of $\bar K$.

Conversely, the measure~$\mu_D$ can equivalently be constructed directly using the above characteristic function and the Bochner--Minlos Theorem, see for example~\cite{JPActa,LSJSP}.

Since the choice of ONB $(u_i)_{i\in\nn}$ and sequence $(l_i)_{i\in\nn}$ is arbitrary, we can choose, for any self-adjoint trace class operator~$a$ with trivial kernel, $(u_i)_{i \in \nn}$ an ONB diagonalizing~$|a|$, and $(l_i)_{i\in\nn}$ its eigenvalues divided by $\tr|a|$. We then have $\sup_{x\in\bar K}|\langle x, a x\rangle|/\|x\|_{\bar K}<\tr|a|$.
Hence $x\mapsto \braket{x,ax}$ extends by continuity to a finite function on $\supp \mu_D\subset \bar K$. This procedure is easily adapted for~$a$ any trace class operator.

Since $v$ is trace class, $V$ can be extended by continuity to $\supp \mu_D \ni x\mapsto \langle x,vx\rangle\in \rr$. Similarly $\supp\mu_D\ni x\mapsto \langle x,v_t x\rangle\in\rr$ is extended by continuity from $K$ to $\supp \mu_D$. On the contrary, the extension of the unperturbed Hamiltonian can take infinite values on~$\supp \mu_D$.

\begin{definition}\label{def:heat_classical_quadratic}
	The probability distribution~$\P_t$ describing the \emph{heat variation} is the law of the random variable~$\Delta Q$ defined in Equation~\eqref{eq:def_DQ_classic} with respect to~$\mu_D$. Namely, $\P_t$ is the law of
	$$
		x \mapsto \braket{x,vx} - \braket{\e^{t\mathcal{L}}x,v\e^{t\mathcal{L}}x},
	$$
	with respect to~$\mu_{D}$.
\end{definition}

\begin{remark}
	Because $v$ and $\e^{t\mathcal{L}^*}v\e^{t\mathcal{L}}$ are trace class, this random variable is $\mu_D$-integrable.
\end{remark}

In this quadratic model, the Fourier transform of $\P_t$ extends analytically to an open neighborhood of $\rr$ in $\cc$. Hence, all the moments of $\Delta Q$ exist. A finer study of the following proposition, including a proof of a large deviation principle with linear rate function can be found in \cite{BJPHarmonic}.
\begin{proposition}\label{prop:class-exp}
	Given $\P_t$ of Definition \ref{def:heat_classical_quadratic}, for all $t\in\rr$, there exists $\gamma_t>0$ such that
	\begin{equation}\label{eq:exponential_bound_classical}
		\ee_t [\e^{\gamma_t |\Delta Q|}] < \infty.
	\end{equation}
	Moreover, if $-1\not\in\sp v$, then there exists $\gamma>0$ such that,
	\begin{equation}\label{eq:uniform_exponential_bound_classical}
		\sup_{t\in\rr_+} \ee_t [\e^{\gamma |\Delta Q|}] < \infty.
	\end{equation}
\end{proposition}
\begin{proof}
Since $\ee_t[\e^{\gamma |\Delta Q|}]\leq \ee_t[\e^{\gamma \Delta Q}]+\ee_t[\e^{-\gamma \Delta Q}]$, the first bound \eqref{eq:exponential_bound_classical} follows directly from properties of Gaussian measures and the fact that $v-\e^{t\mathcal{L}^*}v\e^{t\mathcal{L}}$ is trace class.
The uniform bound \eqref{eq:uniform_exponential_bound_classical} follows from the fact that $-1\not \in \sp v$ implies $v-\e^{t\mathcal{L}^*}v\e^{t\mathcal{L}}$ is uniformly bounded with respect to $t$. Indeed, in that case, $\one+v$ is invertible with bounded inverse so that $\e^{t\mathcal L^*}(\one+v)\e^{t\mathcal{L}}=\one+v$ implies $\|\e^{t\mathcal L}\|^2\leq \|\one+v\|\|(\one+v)^{-1}\|$.
\end{proof}
\begin{remark}
	Note that the size of the interval of~$\gamma\in\rr_+$ for which~\eqref{eq:uniform_exponential_bound_classical} holds is typically decreasing in $\|v\|$. This model then illustrates the expectation stated in the Introduction that the heat fluctuations are controlled by~$\lambda$ when~$V$ is unbounded. Further discussion of the relationship between~$\lambda$ and~$\gamma$ for such harmonic models can be found in~\cite{BJPHarmonic}.
\end{remark}
\begin{corollary}\label{prop:class-moments}
	Given $\P_t$ of Definition \ref{def:heat_classical_quadratic},
	\[
		\ee_t [|\Delta Q|^{k}]<\infty
	\]
	for all $k\in\nn$ and any $t\in\rr_+$. Moreover, if $-1\not\in\sp v$, then
	\[
		\sup_{t \in \rr} \ee_t [|\Delta Q|^{k}] <\infty.
	\]
	for all $k\in\nn$,
\end{corollary}

We will see in Section~\ref{sec:model-bosons} that, for the quantum harmonic oscillator in the context of quantum two-time measurement, further assumptions on~$V$ are needed to prove analogues of Proposition~\ref{prop:class-exp} and Corollary~\ref{prop:class-moments}.
\begin{example}
	We give an example of this abstract setting that can be compared with the quantum harmonic oscillator of Section \ref{sec:model-bosons}. We consider a single harmonic oscillator interacting with a Gaussian 1-dimensional thermal bath at inverse temperature $\beta > 0$. The oscillator has Hilbert space $K_\text{osc} = \rr \oplus \rr$ with inner product
	$$
		\braket{(p,q),(p',q')}_\text{osc} = p p' +  q q'.
	$$
 Let $\dot H^1(\rr)$ be the completion of $H^1(\rr)$ with respect to the norm $\|\nabla{\,\cdot\,}\|$. Then the bath Hilbert space is the vector space $K_\text{bath}:=L^2(\rr) \oplus \dot H^1(\rr)$ equipped with the inner product
	$$
		\braket{(\pi,\phi),(\pi',\phi')}_\text{bath} = \braket{\pi,\pi'}_{L^2} + \braket{\nabla \phi, \nabla \phi'}_{L^2}.
	$$
	The Hilbert space for the compound system is then $K = K_\text{osc} \oplus K_\text{bath}$ with unperturbed Hamiltonian
	\begin{align*}
		H_0 : K_\text{osc} \oplus K_\text{bath} &\to \rr \\
			(p,q,\pi,\phi) &\mapsto \tfrac{1}{2}(\|(p,q)\|_\text{osc}^2 + \|(\pi,\phi)\|_\text{bath}^2) .
	\end{align*}
											
	We take as an initial state the measure $\mu_D$ for the bounded operator $D = \beta^{-1} \one$ on $K$. With $\ones = (0,1,0,0)$ and $\f = (0,0,0,f)$ for some non-zero $f \in H^2(\rr)$, we consider
	$$
		v = \ones \braket{\f,{\cdot}} + \f \braket{\ones, {\cdot}}.
	$$
						The corresponding perturbed Liouvillean is then
	$$
		\cL =
			\begin{pmatrix}
				0 & -1 & 0 & - \braket{f, {\cdot}}_{H^1} \\
				1 & 0 & 0 & 0 \\
				0 & \Delta f & 0 & \Delta \\
				0 & 0 & \one & 0
			\end{pmatrix}.
	$$
	The assumption $-1\not\in \sp v$ is satisfied if and only if $\|\nabla f\|_{L^2} \neq 1$.
\end{example}

\subsection{Linear perturbations}\label{sec:class-linear}
Consider the setting of the previous subsection but let $V:K\mapsto \rr$ instead be a real linear form defined by
$$
	V(x):=\langle f,x\rangle
$$
with $f\in K$. The perturbed Hamiltonian on~$K$ is then
\begin{align*}
	H:
				x&\mapsto \tfrac12 \|x\|^2+\langle f,x\rangle.
\end{align*}
\begin{remark}
With an appropriate choice of gauge, these definitions correspond to the Hamiltonian description of the electromagnetic field $x$ in presence of a charge current encoded in~$f$. The hypothesis $f\in K$ corresponds to a mild UV regularity condition that prevents the apparition of infinite electrostatic energy. It is necessary to properly define the model. Physically, this condition is justified by taking into account that charges are not point like but have some ``volume''~\cite[\S 2.3]{spohn2004dynamics}.

The technical condition $f\in K$ here is analogous to the assumption $f\in\Dom\hat e$ in the quantum counter part to this model discussed in Section~\ref{sec:model-bosons-lin}, and is significantly weaker than the essentially necessary UV conditions that are used to control the tails of~$\P_t$ there.
\end{remark}

It is again easy to show that the flow defined by the perturbed Hamiltonian~$H$ is given by
$$
	\Xi^t(x)= \e^{t\mathcal{L}_0}x+(\e^{t\mathcal{L}_0}-1)f
$$
for any $t\in\rr$ and $x\in K$.
\begin{definition}\label{def:heat_classical_linear}
	The probability distribution~$\P_t$ describing the \emph{heat variation} is the law of the random variable $\Delta Q$ defined in Equation~\eqref{eq:def_DQ_classic} with respect to~$\mu_D$. Namely, $\P_t$ is the law of the random variable
	$$
		x\mapsto \langle f,(1-\e^{t\mathcal{L}_0})(f+x)\rangle
	$$
	with respect to the measure $\mu_D$.
\end{definition}

The following propositions are straightforward consequences of the properties of Gaussian random variables and show that the fluctuations of heat are in some sense trivial in this model.

\begin{proposition}\label{prop:class-linear}
	For any $t\in\rr$, $\P_t$ is a Gaussian probability measure with mean
				$\langle f,(1-\e^{t\mathcal{L}_0})f\rangle$
	and variance
	$\| D^{\frac12}(1-\e^{-t\mathcal{L}_0})f\|^2$.
	\end{proposition}

\begin{proof}
	By definition $\mu_D$ is such that any projection of $x$ on a finite dimensional subspace of $K$ is a centered Gaussian random vector with same dimension as the subspace and covariance given by the projection of $D$ on this subspace. Hence, considering the projection of $x$ on the one dimensional linear subspace spanned by $(\one-\e^{t\cL_0})f$, for any $\alpha\in\rr$,
	$$
		\ee_t(\e^{\i\alpha \Delta Q})=\e^{\i \alpha \langle f,(1-\e^{t\mathcal{L}_0})f\rangle - \frac12 \alpha^2\| D^{\frac12}(1-\e^{-t\mathcal{L}_0})f\|^2}.
	\eqno\qedhere
	$$
\end{proof}

\begin{remark}
	If $D\mathcal{L}_0=\mathcal{L}_0D$, the variance is $2\langle D^{\frac12}f,(1-\operatorname{cosh}(t\mathcal{L}_0))D^{\frac12} f\rangle$.
\end{remark}

\begin{corollary}\label{cor:class-linear}
For any $\gamma>0$,
$$\sup_{t\in\rr} \mathbb E_t(\e^{\gamma|\Delta Q|})<\infty.$$
\end{corollary}
\begin{proof}
The probability measure $\P_t$ being Gaussian, the proposition follows from the fact that the Fourier transform of $\P_t$ is entire analytic, and from the inequalities $|\braket{f,(1-\e^{t\mathcal{L}_0})f}|\leq 2\|f\|^2$, $\|D^{\frac12}(1-\e^{-t\mathcal{L}_0})f\|\leq 2\|D^{\frac12}\|\|f\|$  and $\e^{|x|}\leq \e^{x}+\e^{-x}$.
\end{proof}

\begin{remark}
Even if the Fourier transform of~$\P_t$ is analytic on~$\cc$, the variance is quadratic in~$\|f\|$. Hence this model illustrates again the control of the heat fluctuations by~$\lambda$ as written in the Introduction.
\end{remark}

In Section~\ref{sec:model-bosons-lin}, we will see that both Proposition~\ref{prop:class-linear} and Corollary~\ref{cor:class-linear} do not hold for the similar quantized model in the two-time measurement framework. In the quantum case, $\P_t$ is the probability measure of an inhomogeneous Poisson process, not of a Gaussian random variable. The existence of a uniform bound similar to the one obtained in Corollary~\ref{cor:class-linear}, depends then on the properties of the Poisson process' intensity.

\section{Control of the tails for bounded perturbations}
	\label{sec:sufficient_cond}
	\subsection{Setup}
	\label{sec:fcs-setup}
	  \label{sec:Cstar-desc}

Let us first briefly introduce the operator algebraic definition of quantum dynamical systems that we adopt. We refer the reader to \cite[\S 2--3]{BR1} for a thorough exposition of this mathematical formalism.

A $C^\ast$-dynamical system is a triplet~$(\cO, \tau, \omega)$, where~$\cO$ is a unital $C^*$-algebra, $(\tau^s)_{s \in \rr}$ is a strongly continuous one-parameter group of {$*$-}{auto\-mor\-phisms} of~$\cO$, and~$\omega$ is a state, {i.e.} a positive linear functional on~$\cO$ satisfying $\omega(\one) = 1$.
We furthermore assume~$\omega$ to be faithful and $\tau$-invariant.\footnote{We restrict ourselves to a faithful $\omega$ for simplicity. Up to some technical details, this assumption can be lifted.} Namely, we require that $\omega(A^*A)=0$ implies~$A=0$, and that
$$
	\omega = \omega \circ \tau^s
$$
for all $s \in \rr$. We denote the generator of~$(\tau^s)_{s \in \rr}$ by~$\delta$. From the representation theory of $C^*$-algebras, there exists a Hilbert space~$\cH$, a~$*$-{iso\-mor\-phism} $\pi : \cO \to \cB(\cH)$, and a unit vector~$\Omega \in \cH$ such that
\begin{equation}
	\omega(A) = \braket{\Omega, \pi(A) \Omega}
\end{equation}
for all~$A \in \cO$ and such that the vector~$\Omega$ is cyclic for~$\pi(\cO)$.
This triple $(\cH, \pi, \Omega)$ is called a \emph{GNS representation} of~$\cO$ associated to the state~$\omega$ and is unique up to unitary equivalence; see for example \cite[\S 2.3.3]{BR1}.
On this Hilbert space~$\cH$, there exists a unique self-adjoint operator~$L$ satisfying
\begin{equation}
	L \Omega = 0 \label{eq:LOmega}
\end{equation}
and
\begin{equation}
	\pi(\tau^s(A)) = \Exp{\i s L} \pi(A) \Exp{-\i s L} \label{eq:Ldef}
\end{equation}
	for all~$A \in \cO$ and $s \in \rr$. The self-adjoint operator~$L$ is referred to as the $\omega$-Liouvillean of~$(\tau^s)_{s \in \rr}$, or simply the \emph{Liouvillean}, and is typically not bounded, nor semi-bounded, for thermodynamic systems.

Finally, to a self-adjoint element~$V$ of~$\cO$ we associate a \emph{perturbed dynamics}~$(\tau_V^s)_{s \in \rr}$ generated by~$\delta_V := \delta + \i[V,{\,\cdot\,}]$. Both~$\delta$ and $\delta_V$ are derivations on~$\cO$. We refer to~$V$ as a (bounded) \emph{perturbation} and associate to it a measure that is the central object of this paper. Note that, in the GNS representation, the perturbed dynamics is implemented by $L + \pi(V)$ in the sense that
$$
	\pi(\tau_V^s(A)) = \Exp{\i s (L + \pi(V))} \pi(A) \Exp{-\i s (L + \pi(V))}
$$
for all $A \in \cO$ and all~$s \in \rr$.

\begin{definition}\label{def:EFS}
	The probability distribution~$\P_t$ describing the \emph{heat variation} between times~$0$ and~$t$ associated to~$(\cO, \tau, \omega)$ and the self-adjoint perturbation~$V \in \cO$ is the spectral measure of the operator
	\begin{equation}\label{semi-L}
		L + \pi(V) - \pi(\tau^t_V(V))
	\end{equation}
	with respect to the vector~$\Omega$. Equivalently, it is the unique probability measure with characteristic function
	\begin{equation}\label{eq:def-in-terms-of-cht}
		\cht(\alpha)
		= \braket{\Omega, \Exp{\i \alpha (L + \pi(V) - \pi(\tau^t_V(V)))} \Omega}.
	\end{equation}
	For a non-negative Borel function $g$, we will use $\ee_t(g(\Delta Q))$ to denote the integral
	$
		\int_\rr g(\Delta Q) \d\P_t(\Delta Q)
	$
	with value in~$[0,\infty]$.
\end{definition}

\begin{remark}\label{rk:finite_dim}
	If $\cO$ is finite dimensional, $\P_t$ is the probability distribution of the heat variation as defined by the two-time measurement protocol outlined in the introduction.

	To see this, let $\cO = M_{n\times n}(\cc)$, with an unperturbed dynamics defined in the Heisenberg picture by $\tau^t(A) = \Exp{\i t H_0} A \Exp{-\i t H_0}$ for all $A \in \cO$ and where the Hamiltonian~$H_0$ is a self-adjoint element of~$M_{n\times n}(\cc)$.
	Let~$\omega(A):=\tr(A\hat \omega)$ with $\hat\omega$ a full-rank density matrix that commutes with~$H_0$.
	
	As presented in \cite[\S 4.3.11]{JOPP}\footnote{Section 4.3.11. \textit{The standard representations of}~$\cO$ of the published version corresponds to Section~2.11 of the version available on the \textit{arXiv}.}, an associated GNS representation for this system is known as the  \emph{standard} GNS representation: it has Hilbert space $\cH = M_{n\times n}(\cc)$ equipped with the inner-product $\braket{A,B} = \tr(A^* B)$ and the representation~$\pi$ is given by left matrix multiplication: $(\pi(A))(B)=AB$. The vector representative~$\Omega$ of~$\omega$ is the matrix~$\hat\omega^{1/2}$.
	Then, the Liouvillean is $L = [H_0, {\,\cdot\,}]$
		and Definition~\ref{def:EFS}, for a self-adjoint perturbation $V  \in M_{n\times n}(\cc)$, reads
	\begin{align*}
		\cht(\alpha) &= \braket{\Omega, \Exp{\i \alpha (\pi(H_0) + \pi(V) - \pi(\tau^t_V(V)))} \Exp{-\i \alpha \pi(H_0)} \Omega}\\
			&= \tr( \hat\omega^{1/2} \e^{ \i t (H_0+V)} \e^{\i \alpha (H_0 + V - V)}\e^{-\i t (H_0+V)}\e^{-\i \alpha H_0} \hat\omega^{1/2} ) \\
			&= \tr( \e^{\i \alpha \e^{ \i t (H_0+V)}H_0\e^{-\i t (H_0+V)}}\e^{-\i \alpha H_0} \hat\omega ).
	\end{align*}
	In the first equality we used $[H_0,\hat\omega]=0$ to change $\hat\omega^{1/2}\Exp{-\i\alpha H_0}$ to $\Exp{-\i\alpha H_0}\hat\omega^{1/2}$. The last expression of the characteristic function is the quantum equivalent of \eqref{eq:two_time_classic} for classical systems.

	Using the spectral decomposition $H_0=:\sum_{j} \epsilon_j P_j$ where $\{P_j\}_j$ is a resolution of the identity, we obtain that
	\begin{align*}
		\cht(\alpha) &= \sum_{i,j} \tr(\Exp{\i \alpha \epsilon_j}P_j \e^{-\i t (H_0+V)}\Exp{-\i \alpha \epsilon_i}P_i \hat\omega  \e^{ \i t (H_0+V)}) \\
			&= \sum_{\Delta Q \in \sp H_0 - \sp H_0} \Exp{\i \alpha \Delta Q} \sum_{\substack{i,j \\ \epsilon_j - \epsilon_i = \Delta Q}}\tr( P_j \e^{-\i t (H_0+V)}P_i \hat\omega \e^{ \i t (H_0+V)} )
	\end{align*}
	is the characteristic function of the probability measure
	$$
		\P_t(\{\Delta Q\}) = \sum_{\substack{i,j \\ \epsilon_j - \epsilon_i = \Delta Q}} \tr( P_j \e^{-\i t (H_0+V)}P_i \hat\omega P_i \e^{ \i t (H_0+V)} P_j).
	$$
	Hence, $\P_t$ is indeed the probability measure of the heat variation defined as the result of a two-time measurement protocol of the Hamiltonian~$H_0$~\cite{Kur,Tas,JOPP,EHM,CHT}.

	Relevant models of infinitely extended quantum systems are obtained as the thermodynamic limit of a sequence of such confined systems. As conveyed {e.g.} in~\cite{JOPP},~\cite{JPPP} and~\cite{BJPPP}, mild assumptions ensure that the corresponding sequence of probability measures converges weakly to the measure $\P_t$ of Definition~\ref{def:EFS} on the limiting infinitely extended system. In Appendix \ref{app:TL}, we provide proofs of this weak convergence for the models studied in Section \ref{sec:models}.
\end{remark}

\begin{remark}\label{rem:t-t-theory}
	As we pointed out in Remark \ref{rk:heat_def} for classical systems, the choice of operator whose spectral measure defines $\P_t$ is partly arbitrary.
		But as in the classical case, the choice of $L+\pi(V)-\pi(\tau_V^t(V))$ can be motivated by the entropy balance equation when~$\omega$ is a $(\tau, \beta)$-KMS state.

	Indeed, in the framework of Tomita--Takesaki theory, if~$\omega$ is a $(\tau,\beta)$-KMS state, the Liouvillean is~$L =-\beta^{-1}\log \Delta_\omega$, where~$\Delta_\omega$ denotes the modular operator for the state~$\omega$.
	With $\Delta_{\omega_{-t} | \omega}$ denoting the relative modular operator (non-comutative analogue of the Radon--Nikodym derivative) between the states~$\omega_{-t} := \omega \circ \tau_V^{-t}$ and~$\omega$,
	one then has the identity~\cite{JPPP}
	\begin{equation}
		\beta(L + \pi(V) - \pi(\tau_V^t(V))) = -\log \Delta_{\omega_{-t} | \omega}.
	\end{equation}
	Hence, defining the distribution~$\P_t$ of heat variation as the spectral measure for the operator $L + \pi(V) - \pi(\tau_V^t(V))$ with respect to~$\Omega$ amounts to lifting the classical entropy balance equation~\eqref{eq:ent-bal-classic} to the level of the operators defining the distributions of heat on one side and entropy production on the other. Indeed, upon taking expectation with respect~$\Omega$, the above identity reduces to
	\begin{equation}
		\beta \ee_t(\Delta Q) = -\braket{\Omega, \log \Delta_{\omega_{-t} | \omega} \Omega},
	\end{equation}
	where the right-hand side is identified as the relative entropy between~$\omega$ and~$\omega_{-t}$, exactly as in~\eqref{eq:ent-bal-classic-avg}.
\end{remark}

\begin{remark}
	Definition~\ref{def:EFS} extends naturally to the case of $W^*$-dynamical system~$(\mathfrak{M}, \tau, \omega)$
		together with a bounded self-adjoint perturbation $V \in \mathfrak{M}$. Theorems~\ref{thm:diff} and \ref{thm:an} then also generalise to this case.
	When considering $V$ merely affiliated to the $W^*$-algebra $\mathfrak{M}$, one may in some cases mimic the construction Definition~\ref{def:EFS} and obtain results analogous to the ones for $C^*$-algebras.
		We do not discuss the general theory of heat variation for such unbounded perturbations, but treat some examples in Sections~\ref{sec:model-bosons} and~\ref{sec:model-bosons-lin}.
\end{remark}
 	\subsection{Results}
	\label{sec:results}
	  Our two main general results for $C^*$-dynamical systems with bounded perturbations describe the control of the tails of~$\P_t$. The first one gives a sufficient condition for the existence and uniform boundedness in time of even moments; the second, for the existence  and uniform boundedness of an analytic extension of the Fourier transform to a neighborhood of~$0$ in~$\cc$.

\begin{theorem}\label{thm:diff}
	Let $(\cO, \tau, \omega)$ be a $C^*$-dynamical system, $V \in \cO$ a self-adjoint perturbation, and~$\P_t$ the probability measure of Definition~\ref{def:EFS}. If $V \in \Dom \delta^n$, then
	\[
		\sup_{t \in \rr} \ee_t[\Delta Q^{2n+2}] <\infty.
	\]
\end{theorem}

\begin{theorem}
		In the same setting, if $V \in \Dom\Exp{\i\frac{\gamma}{2}\delta}\cap\Dom\Exp{-\i\frac{\gamma}{2}\delta}$ for $\gamma > 0$, then
	\[
		\sup_{t \in \rr} \ee_t[\e^{\gamma |\Delta Q|}] <\infty.
	\]
	\label{thm:an}
\end{theorem}

\begin{remark}
	A similar result was already presented in \cite{BJPPP}. It was proved via the thermodynamic limit. Below, we provide a proof in the present framework.
\end{remark}

\begin{remark}
The assumptions $V\in\Dom \delta^n(V)$ and $V\in\Dom \e^{\i\frac{\gamma}2\delta}\cap \Dom\e^{-\i\frac\gamma 2 \delta}$ can be reformulated in terms of regularity of the map $t\mapsto \tau^t(V)$. We have $V\in\Dom \delta^n(V)$ if and only if $t\mapsto \tau^t(V)$ is $n$ times norm differentiable and  $V\in\Dom \e^{\i\frac{\gamma}2\delta}\cap \Dom\e^{-\i\frac\gamma 2 \delta}$ if and only if $t\mapsto \tau^t(V)$ admits a bounded analytic extension to the strip $\{z\in\cc : |\Im z|<\frac\gamma 2\}$.
\end{remark}

\subsection{Proofs}

Recall  that, given $V\in \cO$, the one-parameter group $(\tau_V^t)_{t \in \rr}$ is the perturbed dynamics whose generator is given by $\delta_V=\delta+\i[{\cdot},V]$. We start with some useful properties of derivations, the first two being easily proved by induction.

\begin{lemma} \label{generalized-Leibniz-rule}
Let $\delta$ be a derivation on an algebra $\cO$ and
 $A,B\in \Dom(\delta^n)$, $n\in \nn$. Then\[
\delta^n(AB)=\sum_{k=0}^n {n \choose k}
 \delta^{n-k}(A)\delta^k(B).
 \]
\end{lemma}
\begin{lemma} \label{lemma-Vp}
	Let $\delta$ be a derivation on a $C^*$-algebra $\cO$ and let
	$V_t\in \Dom(\delta^n)$ and $\sup_t \norme{\delta^n(V_t)}<\infty$. Then   $V_t^p\in \Dom(\delta^n)$ for all $p\in \nn$ and
	$
		\sup_t \norme{\delta^n(V^p_t)}<\infty.
	$
\end{lemma}

\begin{lemma} \label{lem:unif-bound-delta-n}
  If $V \in \Dom \delta^n$, then $V - \tau_V^t(V) \in \Dom \delta^n$ and
  $$
	  \sup_{t \in \rr} \|\delta^n(V - \tau_V^t(V))\| < \infty.
  $$
\end{lemma}

\begin{proof}
  Using a Dyson expansion of $V-\tau_V^t(V)$, it is easy to show that $V-\tau_V^t(V)\in\Dom(\delta^n)$. It remains to show that the norm of $\delta^n(V-\tau_V^t(V))$ is uniformly bounded in $t$.

  Since $\|\delta^n(V)\| < \infty$ is independent of~$t$, it suffices to show uniform boundedness of~$\|\delta^n(\tau_V^t(V))\|$. We proceed by induction on~$n$, noting that because $\delta$ is a strongly continuous group generator, $V \in \Dom \delta^n$ implies $V \in \Dom \delta^k$ for $k = 1, \dotsc, n-1, n$. For convenience let $c_n := \sup_{t \in \rr} \|\delta^n(\tau_V^t(V))\|$.

  For $n = 1$, by definition of $\tau_V^t$ from $\delta_V = \delta + \i [V,{\cdot}]$,
  \begin{align*}
	  \delta^1(\tau_V^t(V))
	  		   								&
		  = \tau_V^t(\delta(V)) - \i [V,\tau_V^t(V)]
  \end{align*}
  and thus
  \begin{align*}
	  \|\delta^1(\tau_V^t(V))\|
		  &\leq \|\tau_V^t(\delta(V))\| + \|[V,\tau_V^t(V)]\|
		  \leq \|\delta(V)\| + 2\| V\|^2.
  \end{align*}
Hence $c_1\leq  \|\delta(V)\| + 2\| V\|^2$.

  Now suppose that the statement holds for $k$ and that $V \in \Dom \delta^{k+1}$. Then,
  \begin{align*}
	  \delta^{k+1}(\tau_V^t(V))
				  &= \tau_V^t(\delta_V^{k+1}(V)) - \i \sum_{l=0}^k \delta_V^l [V,\delta^{k-l}(\tau_V^t(V))]
  \end{align*}
  and thus, using Lemma~\ref{generalized-Leibniz-rule},
  \begin{align*}
	  \|\delta^{k+1}(\tau_V^t(V))\|
		  &\leq \|\delta_V^{k+1}(V)\| + \sum_{l=0}^k \sum_{j=0}^l {l \choose j} 2\|\delta_V^j(V)\| \| \delta_V^{l-j} \delta^{k-l} (\tau_V^t(V)) \|.
  \end{align*}
  By repeated application of the product rule of Lemma~\ref{generalized-Leibniz-rule}, this may in turn be bounded uniformly in~$t$ by a combination of powers of $c_1, \dotsc, c_k$, $\|V\|, \|\delta(V)\|, \dotsc, \|\delta^k(V)\|$ and $\|\delta^{k+1}(V)\|$, which are all finite by induction hypothesis.
\end{proof}

In the following, $L$ and~$W$ are operators on a Hilbert space, with $L$ possibly unbounded. We denote by~$\ad_L(W)$ the commutator~$[L,W]$. We set $\ad^0_L W=W$. The following formula can be easily proven by induction. We will then be ready to prove Theorems~\ref{thm:diff} and~\ref{thm:an}.

\begin{lemma} \label{lemma-LW}
Let $L$ and $W$ be two linear operators on a Hilbert space. Assume that $\ad^k_L(W)$ extends to a bounded operator for all $k$ with $0\leq k \leq n $. Then,
\[
L^n W= \sum_{k=0}^n \begin{pmatrix}
				n \\
				k\\
			\end{pmatrix}
 \ad_L^{k}(W)L^{n-k}.
\]
\end{lemma}

\begin{proof}[Proof of Theorem~\ref{thm:diff}]
Set $X_t:=\pi(V-\tau_V^t(V))$. By definition of $\mathbb{P}_t$ as a spectral measure, and since $L\Omega=0$, we have
\[
	\int_\rr |\Delta Q|^{2n+2}\d {\mathbb P}_t(\Delta Q) =\norme{(L+ X_t)^{n+1}\Omega}^2 =\norme{\left(L+X_t\right)^{n}X_t\Omega}^2
\]
Now,
\[
	(L+X_t)^nX_t=\sum_{A_i=L,X_t} A_1 A_2 \cdots A_n X_t.
\]
Each term in the sum can be written as
\[
	L^{\alpha_1}X_t^{\beta_1}\cdots L^{\alpha_m}X_t^{\beta_m}X_t
\]
with $\sum_i \alpha_i+\sum_i \beta_i=n$.
Using repeatedly Lemma \ref{lemma-LW} and $L\Omega=0$,
$$(L+X_t)^nX_t\Omega=P_t X_t \Omega$$
with $P_t$ a non-commutative polynomial of variables $\ad_L^kX_t^p$ with $k+p\leq n$.

By Lemmas \ref{lemma-Vp} and \ref{lem:unif-bound-delta-n},  $\norme{\delta^k((V-\tau_V^t(V))^p)}$ is uniformly bounded for all $p$ and all $k\leq n$.
Since $\ad_L^k(X_t^p)=\pi(\delta^k((V-\tau_V^t(V))^p))$, it implies $\ad_L^k(X_t^p)$ is uniformly bounded. Hence $\sup_t\norme{P_t}<\infty$, and $(L+X_t)^{n+1}\Omega=P_t\Omega$ proves the theorem.
\end{proof}

\begin{proof}[Proof of Theorem~\ref{thm:an}]
	By hypothesis, for each $N \in \nn$, the function of $s\in\rr$ defined by the truncated Araki--Dyson series
	\begin{align}
		E_V^{(N)}(s):=\one + \sum_{n =1}^N \i^n \int_0^s  \dotsi \int_0^{s_{n-1}} \pi(\tau^{s_n} (V)) \dotsb \pi(\tau^{s_1} (V))  \d s_n \dotsb \d s_1	\label{eq:Dyson-cstar}
	\end{align}
	has an analytic extension to the complex open neighborhood $S_0 :=\{z \in \cc : |{\Im z}| < \frac{1}{2}\gamma\}$. Moreover,
	$$
		\sup_{N\in\nn} \| E_V^{(N)}(z)\|\leq \Exp{|z|v_0}
	$$
	with $v_0:=\sup_{z\in S_0}\|\tau^z(V)\|$.
	Hence, by the Vitali--Porter convergence theorem and equivalence of the weak and strong analyticity, \eqref{eq:Dyson-cstar} converges on~$S_0$ to a norm-analytic $\pi(\cO)$-valued function $S_0 \ni z \mapsto E_V(z)$ satisfying
	\begin{align*}
		\|E_V(z)\| \leq \Exp{|z|v_0}.
	\end{align*}
	On the other hand, the truncated Araki--Dyson series~\eqref{eq:Dyson-cstar} converges to $\Exp{\i s (L + \pi(V))} \Exp{-\i s L}$ for all $s \in \rr$.
	Then we have
	\begin{align*}
		\Exp{-\i s(L + \pi(V) - \pi(\tau_V^t(V)))}\Omega
			&=  \Exp{-\i s(L + \pi(V) - \Exp{\i t (L + \pi(V))}\pi(V)\Exp{-\i t (L + \pi(V))})} \Omega \\
			&= \Exp{\i t (L + \pi(V))} \Exp{-\i s L} \Exp{\i s(L+\pi(V))} \Exp{-\i t (L + \pi(V))}  \Exp{-\i s (L+\pi(V))}\Omega \\
			&= \Exp{\i t (L + \pi(V))} E_V(-s)^* \Exp{-\i t (L + \pi(V))} E_V(-s) \Omega,
	\end{align*}
	with the operators on the right-hand side admitting analytic extensions as functions of~$s$ to the set~$S_0$.
	Therefore,
	\begin{align*}
		\ee_t[\Exp{ \gamma |\Delta Q|}] &\leq  \ee_t[\Exp{ -\gamma \Delta Q}] + \ee_t[\Exp{\gamma \Delta Q}] \\
			&= \|E_V(-\tfrac\i2\gamma)^* \Exp{-\i t (L + \pi(V))} E_V(-\tfrac\i2\gamma) \Omega\|^2
				\\&\qquad
					+ \|E_V(\tfrac\i2\gamma_0)^* \Exp{-\i t (L + \pi(V))} E_V(\tfrac\i2\gamma) \Omega \|^2 \\
			&\leq  \|\tau_V^t(E_V(-\tfrac{\i}{2}\gamma)^*)\|^2 \|E_V(-\tfrac{\i}{2}\gamma) \|^2 \|\Omega\|^2
				\\&\qquad
			 		+ \|\tau_V^t(E_V(\tfrac{\i}{2}\gamma)^*)\|^2\|E_V(\tfrac{\i}{2}\gamma) \|^2   \|\Omega\|^2  \\
			&\leq 2\Exp{2 \gamma v_0}. \qedhere
	\end{align*}
\end{proof}

\section{Heavy tails in some quasi-free gas models}
   \label{sec:models}
	In the present section, we study some models with $V$ bounded and unbounded where the assumptions of Theorems~\ref{thm:diff} and~\ref{thm:an} are related to explicit UV regularity conditions. We moreover prove some converse implications in each model. The results presented in this section are proved in Section \ref{sec:proofs-results}.

	In Section \ref{sec:model-fermions}, in a model of quasi-free fermions, we show that the $2n+2^{\text{nd}}$ moment of~$\P_t$ exists essentially if and only if $V \in \Dom \delta^n$.

	In Section \ref{sec:model-bosons}, we study a model that can be seen as the bosonic analogue of the previous one or as a quantization of the classical model of Section~\ref{sec:model-classic_quad}. We show that, even if~$V$ is unbounded in this model, results similar to Theorems~\ref{thm:diff} and~\ref{thm:an} hold, with a converse to Theorem~\ref{thm:diff}.

	In Section \ref{sec:model-bosons-lin}, we study a model that can be seen as a quantization of the classical model of Section \ref{sec:class-linear}. First we show that~$\P_t$ is the law of an inhomogeneous Poisson process, in sharp contrast with its classical counterpart for which $\P_t$ is the law of a Gaussian random variable. Building on that result, we again show that, even if~$V$ is unbounded in this model, analogues of Theorems~\ref{thm:diff} and \ref{thm:an}~hold. Moreover, we prove that the implications in both of these theorems have converses.

	\medskip

	Before we give the details on these three models, we introduce common notations. Given a Hilbert space $\ch$, we denote by $\Ga(\ch)$  [resp. $\Gs(\ch)$] the anti-symmetric [resp. symmetric] Fock space
	associated to~$\ch$; by~$a$ and ~$a^*$, the usual annihilation and creation operators there; and by~$\dG(b)$ the second quantization of the one-particle operator~$b$ on~$\ch$. Since it is clear from the context, we use the same letter for the fermionic and the bosonic ones. The expression $a^{\sharp}$  stands either for the creation~$a^*$ or for the annihilation~$a$.
   We denote by~$\varphi$ the corresponding field operators
	$
		\varphi(\psi) := \frac{1}{\sqrt 2}(a(\psi) + a^*(\psi)).
	$

	\subsection{Impurity in a quasi-free Fermi gas}
		\label{sec:model-fermions}

We consider a fermionic impurity interacting with a quasi-free Fermi gas. The corresponding unital $C^*$-algebra $\cO$ is generated by $\{a(\phi) : \phi\in\h\}$ with the one-particle Hilbert space $\h = \cc \oplus L^2(\rr_+, \d e)$. The unperturbed dynamic is given by the extension of
\begin{equation}\label{eq:fermions-dyn}
     \tau^s(a^\sharp(\phi)) = a^\sharp(\Exp{\i s h_0} \phi)
  \end{equation}
for~$\phi \in \h$, where
$
   h_0:=\eno \oplus \hat e
$
on~$\h$, to a group of $*$-automorphisms of $\cO$. Here and in what follows, $(\hat e\phi)(e)=e\phi(e)$ and $\eno > 0$.

The intial state~$\omega$ is taken to be the quasi-free state on $\cO$ generated by the Fermi--Dirac density $T = (1+\Exp{\beta h_0})^{-1}$ for some inverse temperature~$\beta > 0$. Then, $\omega$ is a $(\tau,\beta)$-KMS state. We refer the reader to \cite[\S5]{BR2} or \cite[\S X.7]{RS2} for more details.

The impurity--gas interaction is given by the  bounded perturbation
$$
   V = a^*(\f) a(\ones) + a^*(\ones)a(\f),
$$ where $\f = 0 \oplus f \in \cc \oplus L^2(\rr_+, \d e)$ and $\ones = 1 \oplus 0 \in \cc \oplus L^2(\rr_+, \d e)$. Note that~$V = \dG(v)$ for the one-particle rank two operator $$v = \ones\braket{\f, {\cdot}} + \f\braket{\ones, {\cdot}}.$$

\begin{remark}
   The space $L^2(\rr_+, \d e)$ with the unperturbed one-particle Hamiltonian $\hat e$ is chosen for simplicity. This choice captures the essential features of our problem. Our proofs can be adapted to more general separable Hilbert spaces as long as the unperturbed one-particle Hamiltonian is lower bounded and one works in the representation where it is a multiplication operator.
\end{remark}

\begin{remark}
	If there exists $\gamma>0$ such that $f\in\Dom \e^{\frac12\gamma\hat e}$, the norm equality $\|a(\phi)\|=\|\phi\|$ implies that Theorem~\ref{thm:an} holds for said~$\gamma$.
   If one defines~$V$ using a sharp UV cutoff, namely if there exists $\Lambda>0$ such that $f(e)=0$ for all $e>\Lambda$, then $f\in \Dom \e^{\frac12\gamma}$ for any $\gamma\in\rr_+$ and Theorem \ref{thm:an} holds for arbitrarily large $\gamma$. Next theorem shows that UV regularization conditions are not only sufficient but necessary for the existence of moments of $\P_t$.
\end{remark}

\begin{theorem} \label{thm:fermions-equiv-moments}
	Let $\P_t$ be the probability measure of Definition \ref{def:EFS} for $(\cO,\tau, \omega)$ and $V$ defined in this section. Then, for any $n\in\nn$, the following are equivalent
	\begin{enumerate}[label=(\roman*)]
		\item $ \sup_{t \in \rr} \ee_t[\Delta Q^{2n+2}] <\infty \text{;}$\label{it:unif_bound_moment_fermions}
		\item there exists $ t_1<t_2$ such that
			 $ \int_{t_1}^{t_2} \ee_t[\Delta Q^{2n+2}] \d t <\infty \text{;}$\label{it:integrable_moment_fermions}
		\item  $f\in\Dom \hat e^n$.\label{it:diff_function_fermions}
	\end{enumerate}
\end{theorem}

\begin{remark}\label{rem:reg-fermions-quad}
	In this setup, since $\delta(a^\sharp(\phi))=a^\sharp(\i h_0\phi)$, and $\|a(\phi)\|=\|\phi\|$,  $f\in\Dom \hat e^{n}$ implies $V\in\Dom \delta^n$.
	 Hence the condition $V\in\Dom \delta^n$ is equivalent to an UV regularity condition on the form factor $f$.
\end{remark}

\begin{remark}
	If there exists $n\in\nn$ such that $f\not\in \Dom \hat e^n$, then the $2n+2^{\text{nd}}$ moment of $\P_t$ essentially never exists	{} and $\P_t$ is then heavy-tailed.
\end{remark}

	\subsection{Quantum Open Harmonic Oscillator}
		\label{sec:model-bosons}
		We consider a quantum harmonic oscillator coupled to a gas of quasi-free scalar bosons. Our model is similar to the one studied in \cite{arai1981model,davies1973harmonic}. It is the bosonic analogue of the model of the previous section and the quantum analogue of the classical model of Section~\ref{sec:model-classic_quad}.

Consider the Hilbert  space $$\Gs(\cc)\otimes \Gs(L^2(\rr_+, \d e)) \cong L^2(\rr)\otimes \Gs (L^2(\rr_+, \d e))$$ with unperturbed Hamiltonian
$$
	\dG(\eno)\otimes \one+\one \otimes \dG(\hat{e}).
$$
The oscillator--gas interaction is given by the unbounded perturbation
$$
	a^*(1)\otimes a(f)+  a(1)\otimes a^*(f)
$$
with $f\in \Dom\hat e^{-1/2}\cap \Dom\hat e$.\footnote{The assumption $f\in \Dom\hat e^{-1/2}\cap\Dom \hat e$ is a technical IR/UV regularity condition needed to have a properly defined model.}

Again by the usual identification $\Gs(\cc)\otimes\Gs( L^2(\rr_+, \d e)) \cong \Gs(\cc\oplus L^2(\rr_+, \d e))$, we obtain unitarily equivalents $\dG(h_0)$ and $V := a^*(\f) a(\ones) + a^*(\ones)a(\f)$ on~$\Gs(\h)$ where $\f = 0 \oplus f \in \cc \oplus L^2(\rr_+, \d e)$ and $\ones = 1 \oplus 0 \in \cc \oplus L^2(\rr_+, \d e)$.

Let $\rho = (\Exp{\beta h_0}-1)^{-1}$ be the Bose--Einstein density for some inverse temperature~$\beta > 0$. Let us introduce a representation of the \emph{canonical commutation relations} (CCR) \emph{algebra}  over $\Dom \rho^{1/2}$ known as the (glued left) Araki--Woods representation, which in its earliest form goes back to \cite{AWo}. This representation is determined  by
\begin{align}
	W^{\text{AW}} : \Dom \rho^{1/2} &\to \B(\Gs(\h \oplus \h)) \notag \\
	\phi &\mapsto W(\sqrt{1+\rho}\,\phi  \oplus \sqrt{\rho}\,\overline{\phi}), \label{eq:AWo-rep}
\end{align}
where $W$ denotes the Weyl operators on the symmetric Fock space $\Gs(\h \oplus \h)$ and~$\bar{\phi}$ is the complex conjugation of~$\phi$. These operators satisfy the Weyl form of the CCR:
\begin{equation} \label{eq:weyl}
	W^{\text{AW}}(\phi) W^{\text{AW}}(\psi) = \Exp{-\frac{\i}{2} \Im \braket{\phi, \psi}} W^{\text{AW}}(\phi+\psi) \end{equation}
for all $\phi$ and~$\psi$ in~$\Dom \rho^{1/2}$. Also, we have the evaluation identity
\begin{equation} \label{eq:weyl-eval}
	\braket{\Omega, W^{\text{AW}}(\phi) \Omega} = \Exp{-\frac{1}{4} \braket{\phi, \phi} - \frac{1}{2}\braket{\phi, \rho\phi}}
\end{equation}
for all $\phi$ in~$\Dom \rho^{1/2}$.

\medskip
Let $\mathfrak{M}^{\text{AW}} \subseteq \B(\Gs(\h \oplus \h))$ be the von Neumann algebra generated by $$\{W^{\text{AW}}(\phi) : \phi \in \Dom \rho^{1/2}\}.$$ The maps defined by
\begin{equation}\label{eq:bosons-dyn}
	\tau^s(W^{\text{AW}}(\phi)) = W^{\text{AW}}(\Exp{\i s h_0} \phi)
\end{equation}
for all $\phi \in \Dom \rho^{1/2}$ extend to a $W^*$-dynamics $(\tau^s)_{s \in \rr}$ on $\mathfrak{M}^{\text{AW}}$ {i.e.} a $\sigma$-weakly continuous one-parameter group of $*$-automorphisms of $\mathfrak{M}^\text{AW}$.
We consider the state $\omega : A \mapsto \braket{\Omega, A \Omega}$ on $\mathfrak{M}^{\text{AW}}$ where~$\Omega$ is the vacuum in~$\Gs(\h \oplus \h)$. Then, $\omega$ is a $(\tau,\beta)$-KMS state on $\mathfrak{M}^{\text{AW}}$ and the corresponding Liouvillean is~$L = \dG(h_0 \oplus -h_0)$.
For proofs of these facts, we refer the reader to~\cite[\S 17]{DG} and \cite[\S 9]{DE}.

For each $\phi \in \Dom \rho^{1/2}$, the map $\rr \ni s \mapsto W^\text{AW}(s\phi)$ is strongly continuous and we will denote its generator by~$\varphi^{\text{AW}}(\phi)$. It is explicitly given on the Fock space~$\Gs(\h \oplus \h)$ by
\begin{equation*}
	 \varphi^{\text{AW}}(\phi)
		= \varphi(\sqrt{1+\rho}\,\phi  \oplus \sqrt{\rho}\,\overline{\phi}).
\end{equation*}

The affiliated creation and annihilation operators, ${a^{\text{AW}}}^*(\phi)$ and $a^{\text{AW}}(\phi)$, are related to $\varphi^{\text{AW}}(\phi)$ by the formula $\varphi^{\text{AW}}(\phi) = \tfrac{1}{\sqrt{2}}({a^{\text{AW}}}^*(\phi) + a^{\text{AW}}(\phi))$ and $[a^{\text{AW}}(\phi),{a^{\text{AW}}}^*(\psi)] =\langle \phi,\psi\rangle$. They are explicitly given by the formula
\begin{gather*}
		a^{\text{AW}}(\phi) = a(\sqrt{1+\rho}\,\phi \oplus 0) + a^*(0 \oplus \sqrt{\rho}\,\overline{\phi}).
\end{gather*}

The perturbation~$V$ is represented by an operator~$V^{\text{AW}}$ affiliated to~$\mathfrak{M}^{\text{AW}}$ given by
$$
	V^{\text{AW}}  = {a^{\text{AW}}}^*(\ones) a^{\text{AW}}(\f) + {a^{\text{AW}}}^*(\f) a^{\text{AW}}(\ones),
$$
where $\f = 0 \oplus f \in \Dom \rho^{1/2} \subset \h$ since $f\in\Dom \hat e^{-1/2}$ and $\ones = 1 \oplus 0 \in \h$.

Since $f \in  \Dom \hat{e}^{-1/2} \cap\Dom \hat{e}$, the operator
\begin{equation}
\label{eq:def-L-wwA}
	L+V^{\text{AW}}
\end{equation}
is essentially self-adjoint on $\Dom L \cap \Dom V^{\text{AW}}$ by Lemma~\ref{lemma_selfa_impurity} and the family of maps $(\tau_V^s)_{s \in \rr}$ defined by
\begin{align*}
	\tau^s_V(W^\text{AW}(\phi)) &=  \Exp{\i t(L+V^{\text{AW}})} W^\text{AW}(\phi) \Exp{-\i t(L+V^{\text{AW}})}
\end{align*}
for $\phi \in \Dom \rho^{\frac 12}$ extends to a $W^*$-dynamics on $\mathfrak{M}^\text{AW}$~\cite{DJP}. Again with the one-particle operator $v = \ones\braket{\f, {\cdot}} + \f\braket{\ones, {\cdot}}$, we have
$$
	\tau^s_V(W^\text{AW}(\phi)) = W^\text{AW}(\Exp{\i s (h_0 + v)} \phi)
$$
for all $\phi \in \Dom \rho^{1/2}$.

The operator
$$
	L+V^{\text{AW}}- \Exp{\i t(L+V^{\text{AW}})} V^{\text{AW}}\Exp{-\i t(L+V^{\text{AW}})}
$$
is well-defined and self-adjoint. Its self-ajointness follows from Lemma \ref{lemma_selfa_impurity} and the observation
$$
	L+V^{\text{AW}}- \Exp{\i t(L+V^{\text{AW}})} V^{\text{AW}}\Exp{-\i t(L+V^{\text{AW}})}=\Exp{\i t(L+V^{\text{AW}})} L\Exp{-\i t(L+V^{\text{AW}})}.
$$
Hence the following definition is sound.
\begin{definition}\label{def:FCS-OH}
The distribution~$\P_t$ of heat variation is the spectral measure of
$$ L+V^{\text{AW}}- \Exp{\i t(L+V^{\text{AW}})} V\Exp{-\i t(L+V^{\text{AW}})} $$
with respect to~$\Omega$.
\end{definition}

\begin{theorem} \label{thm:bosons-quad-equiv-moments}
	Assume $f \in \Dom\hat e^{-1/2}\cap \Dom\hat e$ and $\|\hat e^{-1/2} f \| \neq \eno^{1/2}$. Let $\P_t$ be the probability measure of Definition \ref{def:FCS-OH}.	Then, for any $n\in\nn$ the following are equivalent
	\begin{enumerate}[label=(\roman*)]
		\item $ \sup_{t \in \rr} \ee_t[\Delta Q^{2n+2}] <\infty \text{;}$\label{it:unif_bound_moments_bosons_quad}
		\item there exist $ t_1<t_2$ such that $ \int_{t_1}^{t_2} \ee_t[\Delta Q^{2n+2}] \d t <\infty \text{;}$\label{it:integrability_moments_bosons_quad}
		\item  $f\in\Dom \hat e^n$. \label{it:diff_function_bosons_quad}
	\end{enumerate}
\end{theorem}

\begin{remark}
As one can see in the proof, without the assumption $\|\hat e^{-1/2} f \| \neq \eno^{1/2}$, the second statement still implies the third one, and the third statement implies the second one for all time intervals~$(t_1,t_2)$.
\end{remark}

\begin{proposition} \label{prop:bosons-quad-an}
	Let $\P_t$ be the probability measure of Definition \ref{def:FCS-OH}. If there exists $\gamma > 0$ such that $f \in \Dom \Exp{\frac 12 \gamma \hat e}$, then there exists $\gamma'\in (0,\gamma)$ such that
	\begin{equation}
		\sup_{t\in\rr}\ee_t[\e^{\gamma' |\Delta Q|}] < \infty.
	\end{equation}
\end{proposition}

	\subsection{van Hove Hamiltonian}
		\label{sec:model-bosons-lin}
		We consider a quasi-free gas of scalar bosons interacting with a prescribed classical external scalar field. This model is the quantum counterpart to the classical model of Section~\ref{sec:class-linear}.  It typically models the interaction of an electromagnetic field with some prescribed classical charges.

The one-particle Hilbert space is~$\h = L^2(\rr_+, \d e)$ and the one-particle Hamiltonian is~$h_0 = \hat{e}$. The unperturbed Hamiltonian is~$\dG(\hat e)$ on the Fock space~$\G(L^2(\rr_+, \d e))$. The prescribed classical external field induces a perturbation $V:=\varphi(f)$. The perturbated Hamiltonian is
$$
   \dG(\hat{e})+\varphi(f)
$$
with $f\in\Dom \hat e\cap \hat e^{-1/2}$.

Hamiltonians of this type are often called \emph{van Hove Hamiltonians}~\cite{derezinski2003van}. We cast this into the $W^*$-algebraic framework again through the Araki--Woods representation.

Let $\rho$ be the Bose--Einstein density at inverse temperature~$\beta>0$, $\rho:=(\Exp{\beta h_0}-1)^{-1}$. In the corresponding Araki--Woods representation~\eqref{eq:AWo-rep}, the dynamics~$\tau$ is again given by~\eqref{eq:bosons-dyn} and we consider the quasi-free state~$\omega$ associated to~$\rho$. Recall that $\omega$ is a $(\tau,\beta)$-KMS state.
The CCR~\eqref{eq:weyl} and the evaluation identity~\eqref{eq:weyl-eval} still hold.

Since  $f \in \Dom \hat{e}\cap \Dom \hat{e}^{-1/2}$, the van Hove perturbation
\begin{equation}
\label{eq:def-L-vH}
L+\varphi^{AW}(f),
\end{equation}
of the Liouvillean is essentially self-adjoint on $\Dom L \cap \Dom \varphi^{AW}(f)$
by Lemma \ref{lemma_selfa_linear}. The maps
$$
	\tau_V^s(W^\text{AW}(\phi)) = \Exp{\i s (L +  \varphi^\text{AW}(f))} W^\text{AW}(\phi) \Exp{-\i s (L +  \varphi^\text{AW}(f))}
$$
on $\{ W^\text{AW}(\phi) : \phi \in \Dom \rho^{1/2}\}$ extend to a $W^*$-dynamics on $\mathfrak{M}^\text{AW}$.

Then, as in the previous subsection, we can consider the operator
\begin{equation}
L + \varphi^\text{AW}(f) - \Exp{\i t(L + \varphi^\text{AW}(f))}\varphi^\text{AW}(f)  \Exp{-\i t(L + \varphi^\text{AW}(f))},
\end{equation}
and Lemma \ref{lemma_selfa_linear} ensures its essential self-adjointness and soundess of the following definition of the distribution of heat variation.

\begin{definition}\label{def:FCS-vanHove}
The distribution~$\P_t$ of heat variation is the spectral measure of
$$ L + \varphi^\text{AW}(f) - \Exp{\i t(L + \varphi^\text{AW}(f))}\varphi^\text{AW}(f)  \Exp{-\i t(L + \varphi^\text{AW}(f))}$$
with respect to~$\Omega$.
\end{definition}

\begin{remark}\label{rem:reg-bosons-lin}
This system has the particularity that it admits a closed form expression for the characteristic function~$\cht$ of~$\P_t$, which allows an in-depth study of the relation between the tails of~$\P_t$ and the regularity of the perturbation. As in the previous models,  this regularity is expressed in terms of the the map $\rr \ni s \mapsto \Exp{\i s \hat e} f \in L^2(\rr_+, \d e)$, which is $n$ times norm differentiable if and only if $f \in \Dom \hat{e} ^n$.
\end{remark}

\begin{theorem}\label{thm:vanHove_Poisson}
Assume $f \in \Dom \hat{e} \cap \Dom \hat{e}^{-1/2}$. Then, the probability measure~$\P_t$ of Definition \ref{def:FCS-vanHove} is the law of an inhomogeneous Poisson process on~$\rr$ with intensity
$$\d\nu_t(e):=\frac{(1-\cos(et))}{e^{2}}\frac{|f(|e|)|^2}{|1-\e^{-\beta e}|}\d e.$$
\end{theorem}

\begin{remark}
This distribution of the heat variation is in sharp contrast with the Gaussian heat variation of its classical analogue we discussed in Section~\ref{sec:class-linear}.
\end{remark}

\begin{theorem} \label{thm:bosons-lin-equiv-moments}
	Assume $f \in \Dom \hat{e} \cap \Dom \hat{e}^{-1/2}$. Let $\P_t$ be the probability measure of Definition~\ref{def:FCS-vanHove}.
	Then, for any $n\in\nn$ the following are equivalent
	\begin{enumerate}[label=(\roman*)]
		\item $ \sup_{t \in \rr} \ee_t[\Delta Q^{2n+2}] <\infty $;\label{it:uniform_bound_moments_bosons_lin}
		\item there exist $ t_1<t_2$ such that $\int_{t_1}^{t_2} \ee_t[\Delta Q^{2n+2}] \d t <\infty $;\label{it:integrable_moments_bosons_lin}
		\item $f\in\Dom \hat e^n$. \label{it:diff_fun_bosons_lin}
	\end{enumerate}
\end{theorem}

\begin{theorem} \label{thm:bosons-lin-equiv-an}
	Assume $f \in \Dom(\hat{e})\cap \Dom(\hat{e}^{-1/2})$. Let $\P_t$ be the probability measure of Definition~\ref{def:FCS-vanHove}. Then, for $\gamma > 0$, the following are equivalent
	\begin{enumerate}[label=(\roman*)]
		\item $ \sup_{t \in \rr}  \ee_t[\e^{\gamma |\Delta Q|}] <\infty $;\label{it:unif_bound_fourier_bosons_lin}
		\item there exist $ t_1<t_2$ such that
			 $ \int_{t_1}^{t_2} \ee_t[\e^{\gamma |\Delta Q|}] \d t <\infty $;\label{it:integrable_fourier_bosons_lin}
		\item  $f\in\Dom \e^{\frac12\gamma \hat e}$.\label{it:an_fun_bosons_lin}
	\end{enumerate}
\end{theorem}

\section{Proofs for quasi-free gas models}
	\label{sec:proofs-results}\label{sec:models-proofs}
	\subsection{One-particle space preliminaries}
		Some of the results of Section~\ref{sec:models} require technical lemmas on one-particle operators and vectors. Throughout this section, $h_0 = \eno \oplus \hat{e}$ on $\h = \cc \oplus L^2(\rr_+, \d e)$, and the one-particle perturbation~$v$ is of the form $\ones\braket{\f, {\,\cdot\,}} + \f\braket{\ones, {\,\cdot\,}}$ for $\f = 0 \oplus f \in \cc \oplus L^2(\rr_+, \d e)$ and $\ones = 1 \oplus 0 \in \cc \oplus L^2(\rr_+, \d e)$.

First note that~$\Dom h_0 = \Dom h$ and, if~$f \in \Dom \hat{e}^{n-1}$, $\Dom h^n = \Dom h_0^n$. Moreover, we have the following first two lemmas.

The first lemma is a somehow surprising, but quite straight forward, property of the comparison between perturbed and unperturbed one-particle unitary dynamical groups.
\begin{lemma}\label{lem:powers-n-to-diff}
	Let $n \in \nn^*$ and assume $f \in \Dom \hat{e}^{n-1}$. Then $h_0^n(\Exp{\i t h_0} - \Exp{\i t h})$ is bounded and  there exists an affine function $\rr_+\ni t \mapsto K_n(t)$ such that
	\begin{equation}
		\|h_0^n (\Exp{\i t h_0} - \Exp{\i t h})\| \leq K_n(|t|).
	\end{equation}
\end{lemma}

\begin{proof}
	We prove the lemma for $t>0$. The proof for $t<0$ is similar.

	Let $\chi \in\Dom h_0^n $ and $\phi\in\Dom h_0$. Using Duhamel's formula,
	\begin{align*}
	\langle h_0^n\chi,(\e^{\i t h_0}-\e^{\i t h})\phi\rangle&=\i\int_0^t\langle h_0^n \chi, \e^{\i (t-s)h_0}v\e^{\i s h}\phi\rangle \d s,\\
		&=-\int_0^t\langle \chi,(\partial_s e^{\i(t-s)h_0})h_0^{n-1}v\e^{\i s h}\phi\rangle \d s,
	\end{align*}
	Since $\phi\in\Dom h$, the functions $s\mapsto\langle \ones, \e^{\i s h}\phi\rangle$ and $s\mapsto\langle \f, \e^{\i s h}\phi\rangle$ are in $C^1(\rr)$. Similarly, the functions $s\mapsto \langle \chi,\e^{\i(t-s)h_0}h_0^{n-1}\ones\rangle$ and $s\mapsto \langle \chi,\e^{\i(t-s)h_0}h_0^{n-1}\f\rangle$ are in $C^1(\rr)$ using $\chi\in\Dom(h_0^{n})$ for the second function. Integrating by parts,
	\begin{align*}
	&\i\int_0^t\langle \chi,(\partial_s e^{\i(t-s)h_0})h_0^{n-1}v\e^{\i s h}\phi\rangle \d s \\
		&\qquad\qquad			= \i\big[\langle \chi,e^{\i(t-s)h_0}h_0^{n-1}v\e^{\i s h}\phi\rangle\big]_0^t
			+ \int_0^t\langle \chi,e^{\i(t-s)h_0}h_0^{n-1} v h \e^{\i s h}\phi\rangle \d s.
	\end{align*}
	Since $\ones$ and $\f$ are in $\Dom(h_0^{n-1})$,
	$$\|h_0^{n-1}v\|<\infty\quad\text{and}\quad \|h_0^{n-1}vh\|<\infty.$$
	The inequality,
	$$|\langle h_0^n\chi,(\e^{\i t h_0}-\e^{\i t h})\phi\rangle|\leq\|\chi\|\;\|\phi\|\big(2\|h_0^{n-1}v\|+t \|h_0^{n-1}vh\|\big).$$
	and Riesz's Lemma yield the lemma with the affine bound
	$$
		K_n(t)=2\|h_0^{n-1}v\|+t \|h_0^{n-1}vh\| .
	\eqno\qedhere
	$$
\end{proof}

The second Lemma concerns the stability of the domain of powers of $h_0$ with respect to~the perturbed dynamical unitary group.
\begin{lemma}
\label{lem:powers-p}
\label{lemma:sup-t-power-p}
	Let $n \in \nn^*$ and suppose $f \in \Dom \hat{e}^{n-1}$. If $\phi \in \Dom h_0^n$, then for all $\alpha \in [0,n]$,
	\begin{equation}\label{eq:sup-t-power-p}
		\sup_t \|h_0^{\alpha} \Exp{\i t h} \phi\| < \infty.
	\end{equation}
\end{lemma}

\begin{proof}
	By a concavity argument, it suffices to show \eqref{eq:sup-t-power-p} for $\alpha = n$. Since $f \in \Dom \hat{e}^{n-1}$ implies $\Dom h^n = \Dom h_0^n$, the commutation $[\Exp{\i t h}, h^n] = 0$ yields directly that $\Exp{\i t h}\phi \in \Dom h_0^n$.

	Now, by the triangle inequality
	\begin{equation}
		\|h_0^n \Exp{\i t h} \phi \| \leq \|(h_0^n - h^n) \Exp{\i t h} \phi\| + \|h^n \phi\|.
	\end{equation}
	The operator $h_0^{n}-h^n$ is a polynomial in operator variables $h_0^k v^p h_0^{m}$ with $k+p+m\leq n$, $p\geq 1$. Since $f\in\Dom \hat e^{n-1}$,  all the operator variables $h_0^k v^p h_0^m$ are bounded. Hence, $h_0^n-h^n$ is bounded and $\sup_{t}\|h_0^n\e^{\i th}\phi\|<\infty$ follows from $\phi\in\Dom h^n$.
\end{proof}

The last lemma of this subsection allows us to ensure that the existence of moments of $\P_t$ is not the fortuitous consequence of some cycles in the evolution, but truly the consequence of $f\in\Dom \hat e^n$. We will use it in the proof that~\ref{it:integrable_moment_fermions} implies~\ref{it:diff_function_fermions} in Theorems~\ref{thm:fermions-equiv-moments} and~\ref{thm:bosons-quad-equiv-moments}. For that purpose we show that a positive function of~$e$ that will appear later in our proofs is actually lower bounded away from~$0$.

We denote $\bar h:= h_0 +\bar v$ where $\bar v := \ones\langle \overline{\f},\cdot\,\rangle +\overline{\f}\langle\ones,\cdot\, \rangle$.
\begin{lemma}\label{lem:osc-bosons}
Let $A=\sqrt{T}$ or $A=\sqrt{\rho}$. Assume furthermore that $\psi \in \Dom A \setminus \{0\}$ is not an eigenvector of~$h$. Then, for any non-trivial interval $(t_1, t_2)$ of~$\rr$,
	\begin{equation*}
		\inf_{e \in \rr_+} \int_{t_1}^{t_2} \|A (1-\Exp{-\i t (\overline{h}-e)} \psi)\|^2 \d t > 0.
	\end{equation*}
\end{lemma}

\begin{proof}
	Let
	$$F(e):=\left\|\int_{t_1}^{t_2} A (1-\e^{-\i t (\bar h - e)})\psi\d t\right\|.$$
	From triangle inequality on integrals and Jensen's inequality, it is sufficient to lower bound this function on $\rr_+$ to prove the lemma.

	To deal with the unboundedness of $\rho$, we introduce
	$$\rho_{\text{cut}}(e):=\begin{cases}(\e^{\beta e}-1)^{-1}&\mbox{ if }e\geq 1\\ (\e^{\beta}-1)^{-1}&\mbox{ if } e<1.\end{cases}$$
	Let $\rho_1:=(\e^{\beta\eno}-1)^{-1}\oplus \rho_\text{cut}(\hat e)$.	 Since $e\mapsto (\e^{\beta e}-1)^{-1}$ is decreasing, it is a bounded operator and $\rho\geq \rho_1$. Hence lower bounding $e\mapsto F(e)$ with $A=\sqrt{\rho_1}$ implies the existence of the same lower bound for $e\mapsto F(e)$ with $A=\sqrt{\rho}$. Thus we now prove a lower bound away from $0$ on $e\mapsto F(e)$ with $A$ assumed to be a bounded positive operator with trivial kernel.

	Let $y:\rr \to \cc$ be the function defined by
	$$y(x)=\int_{t_1}^{t_2}1-\e^{-\i t x} \d t.$$
	The function $y$ is continuous and bounded by $2(t_2-t_1)$ on $\rr$, and $y(x)=0$ if and only if~$x=0$. Fubini's Theorem implies,
	$$F(e)=\|A y(\bar h -e)\psi\|.$$
	Since $\ker A =\{0\}$, it follows from $y(x)=0 \Leftrightarrow x=0$ that $F(e)=0 \Leftrightarrow \bar h \psi=e\psi$ which contradicts our assumption that $\psi$ is not an eigenvector of $h$. Hence, by continuity  of $F$,  on any compact $E\subset\rr_+$,
	$$\inf_{e\in E} F(e)>0.$$
	Hence $\inf_{e\in \rr_+} F(e)=0$ only if $\liminf_{e\to\infty} F(e)=0$.

	Since $\rr\ni x\mapsto y(x-e)$ converges point wise to $t_2-t_1$ when $e\to\infty$ and $y$ is bounded, by Lebesgue's dominated convergence Theorem,
	$$\lim_{e\to\infty} y(\bar h -e )\psi=(t_2-t_1)\psi$$
	in norm. Hence,
	$$\liminf_{e\to\infty} F(e)=\lim_{e\to\infty}\|Ay(\bar h -e)\psi\|=(t_2-t_1)\|A\psi\|.$$
	The triviality of the kernel of $A$ yields the result.
\end{proof}

 	\subsection{Proofs for the impurity in a quasi-free Fermi gas}
		\label{sec:proof_fermions}
The GNS representation we will work in is known as the (glued left) Araki--Wyss representation. This representation of the CAR was introduced in its earliest form in \cite{AWy}. The Hilbert space for this representation is the anti-symmetric Fock space~$\H = \Ga (\h \oplus \h)$ and the representation is determined by
\begin{gather}
		\pi(a(\phi)) = a(\sqrt{1-T}\phi \oplus 0) + a^*(0 \oplus \sqrt{T}\,\overline{\phi}),
\end{gather}
for all $\phi \in \h$. Finally the vector representative~$\Omega$ of~$\omega$ is the vacuum in~$\Ga(\h \oplus \h)$ and the (unperturbed) Liouvillean is $L = \dG(h_0 \oplus -h_0)$. For proofs of these facts, we refer the reader to~\cite[\S{17}]{DG} and \cite[\S{10}]{DE}. For $\phi \in \Dom h_0$, the operators~$L$ and $\pi(a^\sharp(\phi))$ satisfy the commutation relation
\begin{equation}\label{eq:L-a-comm}
	\i[L, \pi(a^\sharp(\phi))]\Psi = \pi(\delta(a^\sharp(\phi)))\Psi=\pi(a^\sharp(\i h_0 \phi))\Psi,
\end{equation}
for any $\Psi \in \Dom L$,
 as seen from differentiating the relation $$\Exp{\i s L} \pi(a^\sharp(\phi)) \Exp{-\i s L}\Psi = \pi(a^\sharp(\Exp{\i s h_0}\phi))\Psi$$ obtained from \eqref{eq:Ldef} and~\eqref{eq:fermions-dyn}. It follows directly from a repeated application of \eqref{eq:L-a-comm} that
 \begin{equation}\label{eq:equiv_dom_delta_dom_e}
 V\in\Dom(\delta^n) \iff f\in\Dom(\hat e^n).
 \end{equation}

\begin{proof}[Proof of Theorem~\ref{thm:fermions-equiv-moments}]
	The fact that the first statement implies the second one is trivial. The fact that the third one implies the first one is immediate from Theorem~\ref{thm:diff} and the equivalence~\eqref{eq:equiv_dom_delta_dom_e}.

	\bigskip

To prove that the the second statement implies the third one, we proceed by induction on~$n$. The result holds for $n = 0$. Assume it holds for some $n-1 \in \nn$ and suppose that there exists a non-trivial interval~$(t_1, t_2)$ on which $t \mapsto \ee_t(\Delta Q^{2n+2})$ is integrable. Then, $t \mapsto \ee_t(\Delta Q^{2n})$ is also integrable by Jensen's inequality and $f \in \Dom \hat{e}^{n-1}$ by the induction hypothesis. The equivalence~\eqref{eq:equiv_dom_delta_dom_e} yields $V\in\Dom \delta^{n-1}$.
Let $X_t := V - \tau_V^t (V)$. Then Lemma~\ref{lem:unif-bound-delta-n} implies $X_t\in\Dom\delta^{n-1}$ and $\sup_{t\in\rr}\|\delta^{n-1}(X_t)\|<\infty$.

\begin{description}\label{claim:A}
\item[Claim A] $\pi(X_t)\Omega\in\Dom L^n$ for a.e.\footnote{Here and in what follows, we use ``a.e.'' for almost every, with respect to the Lebesgue measure.} $t\in(t_1,t_2)$ and $$\int_{t_1}^{t_2}\|L^n\pi(X_t)\Omega\|^2\d t<\infty.$$
\end{description}

First, we prove $\Omega\in\Dom (L+\pi(X_t))^{n+1}$ for a.e. $t\in(t_1,t_2)$. Let $E\mapsto P(E)$ be the projection valued spectral measure of $L+\pi(X_t)$. Let $\Phi\in\cH$ be such that there exists $\Lambda>0$ such that $P([-\Lambda,\Lambda])\Phi=\Phi$. Then by the Cauchy--Schwarz inequality,
$$|\langle \Phi,(L+\pi(X_t))^{n+1}\Omega\rangle|^2\leq \|\Phi\|^2\|P([-\Lambda,\Lambda])(L+\pi(X_t))^{n+1}\Omega\|^2.$$
The norm $\|P([-\Lambda,\Lambda])(L+\pi(X_t))^{n+1}\Omega\|$ is non-decreasing as a function of $\Lambda$. Then \ref{it:integrable_moment_fermions} implies that for a.e. $t\in(t_1,t_2)$, $\lim_{\Lambda\to\infty}\|P([-\Lambda,\Lambda])(L+\pi(X_t))^{n+1}\Omega\|^2=\ee_t(\Delta Q^{2n+2})$. Hence for a.e. $t\in(t_1,t_2)$,
$$|\langle \Phi,(L+\pi(X_t))^{n+1}\Omega\rangle|^2\leq \|\Phi\|^2\ee_t(\Delta Q^{2n+2})$$
and Riesz's Lemma implies $\Omega\in\Dom (L+\pi(X_t))^{n+1}$.

Now we prove {\bf Claim A}. By the induction hypothesis, Lemma \ref{lem:powers-p} and the equivalence \eqref{eq:equiv_dom_delta_dom_e}, commuting repeatedly $L$ with $\pi(X_t)$, it follows that $\pi(X_t)\Dom L^{n-1}\subset\Dom L^{n-1}$. Hence any non-commutative polynomial in $L$ and $\pi(X_t)$ with at most $n-1$ occurrences of $L$ in each monomial preserves $\Dom L^{n-1}$. Working in a representation of $\cH$ in which $L$ is a multiplication operator, we can develop $(L+\pi(X_t))^{n+1}$. Using $L\Omega=0$ and $\i[L,\pi(\delta^p(X_t))]=\pi(\delta^{p+1}(X_t))$ for $p=0,\ldots,n-2$ we obtain,
\begin{equation}\label{eq:dec-into-lower-and-Q}
	(L+\pi(X_t))^{n+1}\Omega=L^n\pi(X_t)\Omega+Q_t\Omega
\end{equation}
with $Q_t$ a non-commutative polynomial in $\{\pi(\delta^{p}(X_t))\}_{p=0,\ldots,n-1}$. Since $\Omega\in\Dom (L+\pi(X_t))^{n+1}$ and $Q_t\Omega\in\cH$, we obtain $\pi(X_t)\Omega\in\Dom L^n$. Moreover, Lemma \ref{lem:powers-n-to-diff} implies,
$$\sup_{t\in\rr}\|Q_t\Omega\|<\infty.$$
Hence, {\bf Claim A} follows from \eqref{eq:dec-into-lower-and-Q}, the triangle inequality, assumption~\ref{it:integrable_moment_fermions} and Jensen's inequality.

\begin{description}
\item[Claim B] $f\in\Dom\hat e^n$
\end{description}
	Using the explicit expressions in Fock space, the non-trivial part of $L^n \pi(X_t) \Omega$ lives in the two-particle subspace of~$\Ga(\h \oplus \h)$ only, namely $P_2L^n\pi(X_t)\Omega=L^n\pi(X_t)\Omega$ with $P_2$ the projector onto the two-particle subspace of~$\Ga(\h\oplus\h)$. The vector there is
			{\footnotesize
			\begin{align}
		 	&(\sqrt{\one - T}  h_0^n \ones \oplus 0)  \wedge (0 \oplus \sqrt{T}  \, \overline{\f}) - (\sqrt{\one - T} h_0^n \Exp{\i t h} \ones \oplus 0)  \wedge (0 \oplus \sqrt{T} \Exp{-\i t\overline{h}} \overline{\f}) \label{eq:2part_vect_1} \\
				&\quad {} + (\sqrt{\one - T}  \ones \oplus 0)  \wedge (0 \oplus \sqrt{T}  \, h_0^n\overline{\f})+ (\sqrt{\one - T}  \f \oplus 0)  \wedge (0 \oplus \sqrt{T}  h_0^n \ones) \label{eq:2part_vect_2}\\
				&\quad {} - (\sqrt{\one - T} \Exp{\i t h} \f \oplus 0)  \wedge (0 \oplus \sqrt{T} h_0^n \Exp{-\i t \overline{h}} \ones) - (\sqrt{\one - T} \Exp{\i t h} \ones \oplus 0)  \wedge (0 \oplus \sqrt{T} h_0^n \Exp{-\i t \overline{h}} \overline{\f}) \label{eq:2part_vect_3}\\
				&\quad {}+ (\sqrt{\one - T}  h_0^n \f \oplus 0)  \wedge (0 \oplus \sqrt{T} \ones)- (\sqrt{\one - T} h_0^n \Exp{\i t h} \f \oplus 0)  \wedge (0 \oplus \sqrt{T} \Exp{-\i t \overline{h}} \ones).\label{eq:2part_vect_4}
			\end{align}
			}
Since $\psi_s\in\Dom h_0^n$, Lemma \ref{lemma:sup-t-power-p} implies the norm of \eqref{eq:2part_vect_1} in $(\h\oplus\h)\wedge(\h\oplus\h)$ is uniformly bounded in $t\in\rr$. Since $\sqrt{T}h_0^n$ is bounded, the norm of \eqref{eq:2part_vect_2}$+$\eqref{eq:2part_vect_3} is also uniformly bounded in $t\in\rr$.
	Now, since, by {\bf Claim A}, $L^n\pi(X_t)\Omega\in\cH$ for a.e. $t\in(t_1,t_2)$, \eqref{eq:2part_vect_4} has finite norm for a.e. $t\in(t_1,t_2)$.
	Recall that Lemma~\ref{lem:powers-n-to-diff} implies $\|h_0^n(\e^{\i t h_0}-\e^{\i t h})\|\leq K_n(|t|)$ with $\rr_+\ni t\mapsto K_n(t)$ an affine function. Because of this and the fact that $\sqrt{1-T}$ is bounded, there exists a family $(\phi_t)_{t\in\rr}$ in $(\h\oplus\h)\wedge(\h\oplus\h)$ with square integrable norm on $(t_1,t_2)$  and such that \eqref{eq:2part_vect_4} is equal to
	\begin{equation}\label{eq:2part_vect_5}
	\begin{split}
	 &(\sqrt{\one - T}  h_0^n \f \oplus 0)  \wedge (0 \oplus \sqrt{T} \ones) \\&\qquad\qquad{} - (\sqrt{\one - T} h_0^n \Exp{\i t h_0} \f \oplus 0)  \wedge (0 \oplus \sqrt{T} \Exp{-\i t \overline{h}} \ones) +\phi_t.
 	\end{split}
	\end{equation}
Since the norms of \eqref{eq:2part_vect_1}, \eqref{eq:2part_vect_2} and \eqref{eq:2part_vect_3} are uniformly bounded in time, and since~$\|\phi_t\|^2$ is integrable on $(t_1,t_2)$, using the triangle inequality, {\bf Claim A} and Jensen's inequality, it follows that $t\mapsto \|\text{\eqref{eq:2part_vect_5}}-\phi_t\|^2$ is integrable on $(t_1,t_2)$:
	\begin{align*}
		&\int_{t_1}^{t_2} \|(\sqrt{\one - T}  h_0^n \f \oplus 0)  \wedge (0 \oplus \sqrt{T} \ones) \\
		&\qquad\qquad\qquad\qquad {}- (\sqrt{\one - T} h_0^n \Exp{\i t h_0} \f \oplus 0)  \wedge (0 \oplus \sqrt{T} \Exp{-\i t \overline{h}} \ones)\|^2 \d t < \infty.
	\end{align*}
	Calculating explicitly the norm squared,
	\begin{equation}
		\int_{t_1}^{t_2} \int_{\rr_+} \left| \frac{1}{\Exp{-\beta e} + 1}\right| e^{2n} |f(e)|^2 \|\sqrt{T}(\one - \Exp{-\i t(\overline{h}-e)}) \ones \|^2 \d e \d t< \infty.
	\end{equation}
	By Fubini's theorem,
	\begin{equation}
		 \int_{\rr_+}  \left| \frac{1}{\Exp{-\beta e} + 1}\right| e^{2n} |f(e)|^2 \int_{t_1}^{t_2} \|\sqrt{T}(\one - \Exp{-\i t(\overline{h}-e)}) \ones \|^2  \d t \d e < \infty.
	\end{equation}
	Lemma~\ref{lem:osc-bosons} then implies
	$$
		\int_{\rr_+} e^{2n} |f(e)|^2  \d e < \infty. \eqno\qedhere
	$$
\end{proof}

 	\subsection{Proofs for the open harmonic oscillator}
		We begin with a few remarks that are immediate from the definitions. One computes directly from the definitions that for any~$\phi \in \Dom h_0 \cap \Dom h_0^{-1/2}$,
\begin{equation}
	\i[L,a^\text{AW}(\phi)]\Psi = a^\text{AW}(\i h_0 \phi)\Psi \label{eq:comm-L}
\end{equation}
for all $\Psi \in \cH_{\rm fin}\cap \Dom L$ with $\cH_{\rm fin}$ the subspace of $\Gs(\h\oplus\h)$ spanned by the vectors with a fixed number of particles: $$\cH_{\rm fin}:=\operatorname{linspan}\{a^*(\phi_n)\dotsb a^*(\phi_1)\Omega : n\in\nn,\ \phi_j\in \h\oplus \h, j=1,\dotsc,n\}.$$

\bigskip
In order to prove Theorem~\ref{thm:bosons-quad-equiv-moments}, we will need the following lemma to control the eventual buildup of infrared irregularities.
\begin{lemma} \label{lemma:sup-t-power-half}
	Assume $\|\hat{e}^{-1/2} f\| \neq \eno^{1/2}$  and $\phi \in \Dom h_0^{-1/2}$, then
	$$
	\sup_{t \in \rr} \| h_0^{-1/2} \Exp{\i t h} \phi \| < \infty.
	$$
\end{lemma}

\begin{proof}
	Since $\ker h_0 = \{0\}$, $h_0^{\pm 1/2} h_0^{\mp 1/2} = \one$ and
	$$
		h_0^{-1/2} h h_0^{-1/2} = \one + (\eno^{-1/2} \ones) \braket{h_0^{-1/2} \f,{\,\cdot\,}} + (h_0^{-1/2} \f)  \braket{\eno^{-1/2} \ones,{\,\cdot\,}}.
	$$
	Since $\psi_f \in \Dom h_0^{-1/2}$, $h_0^{-1/2} h h_0^{-1/2}$ is bounded  and so is $h_0^{-1/2} h^{1/2}$.
	The operator $h_0^{-1/2}hh_0^{-1/2}$ has spectrum
	\begin{equation}
	\label{eq:sph}
		\sp(h_0^{-1/2} h h_0^{-1/2}) = \{ 1, 1 \pm \eno^{-1/2} \|\hat{e}^{-1/2} f \| \}.
	\end{equation}

	Suppose now that there exists~$\psi\in \ker h\setminus\{0\}$. Then $\psi\in\Dom h = \Dom h_0$, so $\psi\in\Dom h_0^{1/2}$.
	Then $\psi'=h_0^{1/2}\psi$ is an eigenvector of $h_0^{-1/2}hh_0^{-1/2}$ of eigenvalue $0$. Since, $\|\hat{e}^{-1/2} f\| \neq \eno^{1/2}$ implies $0$ is not an eigenvalue of $h_0^{-1/2}hh_0^{-1/2}$, the hypothesis~$\|\hat{e}^{-1/2} f\| \neq \eno^{1/2}$ implies $\ker h=\{0\}$ and $h^{\pm 1}h^{\mp 1}=\one$.

	 Hence, $\|\hat{e}^{-1/2} f\| \neq \eno^{1/2}$ implies that $(h_0^{-1/2}hh_0^{-1/2})^{-1}=h_0^{1/2}h^{-1}h_0^{1/2}$ is bounded, and so is $h_0^{1/2}h^{-1/2}$.

	 The inequality
	$$
		|\braket{\psi, h^{-1/2} \phi}| \leq \|h_0^{1/2} h^{-1/2}\| \|\psi\| \|h_0^{-1/2} \phi\|.
	$$
	for all $\psi \in \Dom h_0^{-1/2}$ shows that $\phi \in \Dom h^{-1/2}$.
	Therefore,
	$$
	 \|h_0^{-1/2} \Exp{\i t h} \phi\| = \|h_0^{-1/2} h^{1/2} \Exp{\i t h} h^{-1/2} \phi\| \leq \|h_0^{-1/2} h^{1/2}\| \|h^{-1/2} \phi\|
	 $$
	 which is independent of $t$.
\end{proof}
\begin{remark}
As shown in the proof, assumption $\|\hat e^{-1/2}f\|\neq \eno^{1/2}$ implies $\ker h=\{0\}$ which might be a more recognizable assumption to the reader. Assuming $f\in\Dom(\hat e^{-1})$, the implication is an equivalence. If $f\not\in\Dom(\hat e^{-1})$ and $\|\hat e^{-1/2}f\|=\eno^{1/2}$, $0$ may be in the singular continuous spectrum of $h$ (see~\cite{WWA}).
\end{remark}

\begin{proof}[Proof of Theorem~\ref{thm:bosons-quad-equiv-moments}]
	The fact that the first statement implies the second one is trivial.

\ref{it:integrability_moments_bosons_quad}$\implies$\ref{it:diff_function_bosons_quad}: As for the fermionic case, we proceed by induction on~$n$. The implication is true  for $n = 0$. Assume it holds for some $n-1 \in \nn$ and suppose that there exists a non-trivial interval~$(t_1,t_2)$ on which $t \mapsto \ee_t(\Delta Q^{2n+2})$ is integrable. Then, $t \mapsto \ee_t(\Delta Q^{2n})$ is integrable by Jensen's inequality and $f \in \Dom \hat{e}^{n-1}$ by induction hypothesis.
Consider the set
	\begin{align*}
			\mathcal{G}_{n-1}&:= \{(\rho+1)^{1/2} h_0^p \Exp{\i s h} \f \oplus 0, (\rho+1)^{1/2}  h_0^p \Exp{\i s h} \ones \oplus 0,\\
			 &\qquad \qquad 0 \oplus \rho^{1/2} (-h_0)^p \Exp{-\i s \overline{h}} \overline{\f} , 0 \oplus \rho^{1/2} (-h_0)^p \Exp{-\i s \overline{h}} \ones\}_{s\in\rr;\ p=0,\ldots,n-1}
	\end{align*}
	of vectors in $\h\oplus\h$. Lemmas~\ref{lemma:sup-t-power-p} and \ref{lemma:sup-t-power-half} ensure that
	\begin{equation}\label{eq:uniform_bound_cG}
	m:=\sup_{\phi\in\mathcal{G}_{n-1}}\|\phi\|<\infty
	\end{equation}
	We set  $X_t := V^\text{AW} - \Exp{\i t (L + V^\textnormal{AW})}V^\text{AW}\Exp{-\i t (L + V^\textnormal{AW})}$.
\begin{description}
\item[Claim A] $X_t\Omega\in\Dom L^n\cap\cH_{\rm fin}$ for a.e. $t\in(t_1,t_2)$ and $$\int_{t_1}^{t_2}\|L^nX_t\Omega\|^2\d t<\infty.$$
\end{description}

Since both $L$ and $L+X_t$ are self-adjoint, the proof uses the same arguments as the ones used in the proof of {\bf Claim A} in the proof of Theorem~\ref{thm:fermions-equiv-moments}. The only difference lays in the proof that $\sup_{t\in\rr}\|Q_t\Omega\|<\infty$. We no longer have that $\pi(\delta^p(X_t))$ is uniformly bounded in norm for $p=0,\ldots, n-1$. Instead, using
	\begin{equation}\label{eq:bound_vector_aOmega}
	\|a^\sharp(\phi_k)\ldots a^\sharp(\phi_1)\Omega\|\leq \sqrt{(k+1)!}\prod_{i=1}^k\|\phi_i\|
	\end{equation}
	for $(\phi_i)_{i=1}^k\subset \ch\oplus\ch$, the unifrom bound \eqref{eq:uniform_bound_cG} leads to
	\begin{equation}\label{eq:uniform_bound_QOmega}
	\sup_{t\in\rr}\|Q_t\Omega\|<\infty.
	\end{equation}
	We leave the complete adaptation of the proof to the interested reader.
\begin{description}
\item[Claim B] $f\in\Dom \hat e^n$.
\end{description}

	 Using the explicit expressions in Fock space $L^n X_t \Omega$ lives in the two-particle subspace of~$\Ga(\h \oplus \h)$ only. As in the proof of Theorem \ref{thm:fermions-equiv-moments}, using \eqref{eq:uniform_bound_cG}, \eqref{eq:bound_vector_aOmega} and Lemma \ref{lem:powers-n-to-diff} we deduce from {\bf Claim A} that,
	 \begin{align*}
		&\int_{t_1}^{t_2} \| ((\rho+1)^{1/2} h_0^{n}\f \oplus 0) \otimes_{\rm s} (0 \oplus \rho^{1/2} \ones) \\
		&\qquad\qquad\qquad{} - ((\rho+1)^{1/2} h_0^{n} \e^{\i th_0}\f \oplus 0) \otimes_{\rm s} (0 \oplus \rho^{1/2}\e^{-\i t\bar h} \ones) \|^2 \d t < \infty.
	\end{align*}
	But because the two terms arising from the symmetrized tensor product are orthogonal, we have
	\begin{align*}
		&\int_{t_1}^{t_2} \|(\rho+1)^{1/2} h_0^{n}\f\otimes \rho^{1/2} \ones - (\rho+1)^{1/2} h_0^{n} \e^{\i th_0}\f \otimes \rho^{1/2}\e^{-\i t\overline{h}} \ones\|^2 \d t < \infty.
	\end{align*}
	Therefore, using the spectral representation of~$h_0$,
	\begin{align*}
		\int_{t_1}^{t_2} \int_0^\infty (\rho(e) + 1) |f(e)|^2 e^{2n} \|\rho^{1/2} (\one - \Exp{-\i t (\overline{h}-e)})\ones\|^2 \d e \d t < \infty.
	\end{align*}
	By Fubini's theorem, $\rho(e)+1\geq1$ and Lemma~\ref{lem:osc-bosons}, we conclude
	\begin{equation}
		\int_{0}^{\infty} e^{2n} |f(e)|^2  \d e < \infty
	\end{equation}
	and {\bf Claim B} holds for $n$. By induction, \ref{it:integrability_moments_bosons_quad} implies~\ref{it:diff_function_bosons_quad}.

\medskip
	\ref{it:diff_function_bosons_quad}${\implies}$\ref{it:unif_bound_moments_bosons_quad}: Note that Lemmas~\ref{lemma:sup-t-power-p} and \ref{lemma:sup-t-power-half}  ensures that if the third statement holds, then
	\begin{equation}\label{eq:uniform_bound_cG_n}
	m := \sup_{\phi \in \mathcal{G}_n} \|\phi\|  < \infty
	\end{equation}
	Developing the power of $L+X_t$, using $L\Omega=0$ and $\i[L,{a^{AW}}^\sharp(\phi)]={a^{AW}}^\sharp(\i h_0\phi)$ for any $\phi\in\Dom h_0$ we obtain $(L+X_t)^{n+1}\Omega=Q_t\Omega$ with $Q_t$ a non-commutative polynomial in $a^\sharp(\phi),\ \phi\in \mathcal{G}_n$. It follows then from, \eqref{eq:bound_vector_aOmega} and \eqref{eq:uniform_bound_cG_n} that
	$$\sup_{t\in\rr}\|Q_t\Omega\|<\infty.
	 \eqno\qedhere
	$$
\end{proof}
\begin{remark}
	If we restrict $\mathcal{G}_n$ to $t<T$, then following the same proof, $\sup_{t\in[0,T]}\ee_t(\Delta Q^{2n+2})<\infty$. In that case the assumption $\|\hat e^{-1/2}f\|\neq \eno^{1/2}$ is not necessary.
	Indeed, one can use $f\in\Dom \hat e^{-1/2}$ and $\phi\in\Dom h_0^{-1/2}$ implies $t\mapsto \|h_0^{-1/2}\e^{\i th}\phi\|$ is bounded by a continuous function instead of Lemma~\ref{lemma:sup-t-power-half}. This implication is easily proved using Riesz's Lemma and Duhamel's formula.
\end{remark}
We turn to the proof of Proposition \ref{prop:bosons-quad-an}. We first establish two lemmas.

\begin{lemma}\label{lemma:max-exp-x}
		If there exists $\gamma > 0$ such that $f \in \Dom \Exp{\frac 12 \gamma \hat{e}}$, then for any $\psi\in\Dom \e^{\frac12\gamma h_0}$, and any $\gamma'\in{[{0},{\frac 12}\gamma)}$,
		$$\sup_{t\in\rr}(\|\rho^{1/2}\e^{\gamma' h_0}\e^{\i t h}\psi\| +\|(1+\rho)^{1/2}\e^{\gamma' h_0}\e^{\i t h}\psi\|)<\infty.$$
\end{lemma}

\begin{proof}
	First, let us show that $\e^{\pm\frac12 \gamma h_0}\e^{\mp\frac12 \gamma h}$ is bounded. Let $\phi\in\Dom \e^{\pm \frac12 \gamma h_0}$ and $\phi'\in\Dom \e^{\mp \frac12 \gamma h}$. Then, by repeated differentiation,
	\begin{align*}
	&\langle\phi,\e^{\pm\frac12\gamma h_0}\e^{\mp\frac12\gamma h}\phi'\rangle
	\\&\qquad
		=\sum_n (\mp1)^n \int_{0<\alpha_1<\cdots<\alpha_n<\tfrac12\gamma}\langle\phi,\e^{\pm\alpha_1 h_0}v\e^{\mp\alpha_1 h_0}\dotsb \e^{\pm\alpha_n h_0}v\e^{\mp \alpha_nh_0}\phi'\rangle \quad\\
		&\specialcell{\hfill \d\alpha_1\ldots \d\alpha_n.}
	\end{align*}
	The series on the right-hand side converges uniformly in $\phi$ and $\phi'$ since assumption $\psi_f\in\Dom \e^{\frac12\gamma h_0}$ implies
	$$
		R(\gamma) := \sup_{\alpha\in[-\frac12\gamma,\frac12\gamma]}\|\e^{\alpha h_0}v\e^{-\alpha h_0}\| < \infty.
	$$
	Then, the Cauchy--Schwarz inequality implies
	$$
		|\langle\phi,\e^{\pm\frac12\gamma h_0}\e^{\mp\frac12\gamma h}\phi'\rangle|\leq \|\phi\|\|\phi'\|\e^{\frac12\gamma R(\gamma)}.
	$$
	Thus $\e^{\pm\frac12 \gamma h_0}\e^{\mp\frac12 \gamma h}$ is bounded by $\e^{\frac12 \gamma R(\gamma)}$.

	\medskip

Inserting the identity $\e^{-\frac12 \gamma h}\e^{\frac12 \gamma h}$ we obtain,
	$$\|\e^{\frac12 \gamma h_0}\e^{\i t h}\psi\|=\|\e^{\frac12 \gamma h_0}\e^{-\frac12 \gamma h}\e^{\i t h}\e^{\frac12 \gamma h}\psi\|\leq \e^{\frac12\gamma R(\gamma)}\|\e^{\frac12 \gamma h}\psi\|.$$
	Inserting the identity $\e^{-\frac12 \gamma h_0}\e^{\frac12 \gamma h_0}$, we obtain,
	$$\|\e^{\frac12 \gamma h_0}\e^{\i t h}\psi\|\leq \e^{\frac12\gamma R(\gamma)}\|\e^{\frac12 \gamma h}\e^{-\frac12 \gamma h_0}\e^{\frac12 \gamma h_0}\psi\|\leq \e^{\gamma R(\gamma)}\|\e^{\frac12 \gamma h_0}\psi\|$$
	and we conclude
	\begin{equation}\label{eq:exp-evol-f}
		\sup_{t\in\rr}\|\e^{\frac12 \gamma h_0}\e^{\i t h}\psi\|<\infty.
	\end{equation}

	\medskip
	Since $e\e^{\frac12 \gamma' e}\leq \frac{2}{\epsilon}\e^{\frac12(\gamma' +\epsilon)e}$ for any $\gamma>0$, $\epsilon>0$ and $e\geq 0$, for any $x\in [0,\tfrac12[$,
	$$h_0\e^{x \gamma h_0}\leq \frac{2}{(\frac12 - x)\gamma}\e^{\frac12 \gamma h_0}.$$
	Then, since $\sqrt{\rho} h_0$ is bounded,
	$$\|\sqrt{\rho} \e^{x\frac12 \gamma h_0} \psi\|\leq \|\sqrt{\rho} h_0\|\frac{2}{(\frac12-x)\gamma}\|\e^{\frac12 \gamma h_0}\e^{\i t h}\psi\|.$$
	It then follows from~\eqref{eq:exp-evol-f} that
	$$
		\sup_{t \in \rr} \|\rho^{1/2}\e^{x \gamma h_0}\e^{\i t h}\psi\| < \infty.
	$$
	Similarly
	$$
		\sup_{t \in \rr} \|(\one +\rho)^{1/2}\e^{x \gamma h_0}\e^{\i t h}\psi\| < \infty.
	\eqno\qedhere
	$$
\end{proof}

\begin{lemma} \label{lem:reaminder-for-real}
	If $f \in \Dom \hat{e}\cap\Dom \hat e^{-1/2}$, then there exists $s_0>0$ such that
	\begin{equation*}
		 \sup_{t\in\rr}\left\|\Exp{\i s (L + X_t)}\Omega - \Omega + \sum_{n=1}^N \i^n \int_0^s  \dotsi \int_0^{s_{n-1}}\tau^{s_n} (X_t) \dotsb \tau^{s_1} (X_t) \Omega \d s_n \dotsb \d s_1\right\|
	\end{equation*}
	converges to~0 as~$N \to \infty$, for all $s\in(-s_0,s_0)$.
\end{lemma}

\begin{proof}
	Since $L\Omega=0$, $\Exp{\i s (L + X_t)}\Omega=\Exp{\i s (L + X_t)}\Exp{-\i s L}\Omega$ and for  $\Psi\in\cH_{\rm fin}$,
	$$
	 \partial_s \Exp{\i s (L + X_t)}\Exp{-\i s L} \Psi = \i \Exp{\i s (L+X_t)} \Exp{-\i s L} \tau^s (X_t)\Psi.
	$$
	Since $f\in\Dom \hat e^{-1/2}$ implies $\tau^{s_n} (X_t) \dotsb \tau^{s_1} (X_t) \Omega\in\cH_{\rm fin}$, a repeated application of the fundamental theorem of calculus yields that for any $N \in \nn$,
	\begin{equation}
	\begin{split}
		&\Exp{\i s (L + X_t)} \Exp{-\i s L}\Omega \\
			&\qquad = 1 + \sum_{n=1}^N \i^n \int_0^s \dotsi \int_0^{s_{n-1}} \tau^{s_n} (X_t) \dotsb \tau^{s_1} (X_t) \Omega  \d s_n \dotsb \d s_1 + R_{N}
	\end{split}
	\end{equation}
	with $R_N$ being
	$$
		\i^{N+1} \int_0^s  \dotsi \int_0^{s_{N}}\Exp{\i s_{N+1} (L + X_t)} \Exp{-\i s_{N+1} L} \tau^{s_{N+1}} (X_t) \dotsb \tau^{s_1} (X_t) \Omega  \d s_{N+1} \dotsb \d s_1.
	$$
	Then,
	\begin{align*}
		\|R_N\| &\leq \frac{|s|^{N+1}}{(N+1)!} \sup_{s_1, \dotsc, s_{N+1} \in [0,s]} \|\tau^{s_{N+1}} (X_t) \dotsb \tau^{s_1} (X_t) \Omega\|.
	\end{align*}
	Let
	\begin{align*}
	&\mathcal{G}_0:=\{(\rho+1)^{1/2} \Exp{\i s h_0}\Exp{\i t h} \f \oplus 0, (\rho+1)^{1/2} \Exp{\i s h_0} \Exp{\i t h} \ones \oplus 0,\\
			 &\qquad \qquad 0 \oplus \rho^{1/2} \Exp{\i s h_0}\Exp{-\i t \overline{h}} \overline{\f} , 0 \oplus \rho^{1/2} \Exp{\i s h_0}\Exp{-\i t \overline{h}} \ones\}_{s\in\rr, t\in\rr}.
	\end{align*}
	Lemma \ref{lemma:sup-t-power-half} implies
	\begin{equation}\label{eq:uniform_bound_cG_0}
	m:=\sup_{\phi\in\mathcal{G}_0}\|\phi\|<\infty.
	\end{equation}
	The operators $\tau^s (X_t)$ are non-commutative polynomials of degree $2$ in $a^\sharp(\phi), \phi\in\mathcal{G}_0$. They each consists of 16 monomials. It thus follows from \eqref{eq:bound_vector_aOmega} that
\begin{align*}
		\|R_N\| &\leq \frac{\sqrt{(2N+1)!}}{(N+1)!} (16m^2|s|)^{N+1}
	\end{align*}
	Since $\limsup_{N\to\infty} 2^{-N}\frac{\sqrt{(2N+1)!}}{(N+1)!}<\infty$, $\lim_{N\to\infty} \|R_N\|=0$ for $|s|<\frac1{32m^2}.$
\end{proof}

\begin{proof}[Proof of Proposition~\ref{prop:bosons-quad-an}.]
	Let $F_{t,N} : \rr \to \mathbb{C}$ be defined by the truncated series
	$$
		F_{t,N}(\alpha) := 1 + \sum_{n=1}^N \i^n \int_0^\alpha  \dotsi \int_0^{\alpha_{n-1}}
		\braket{\Omega, \tau^{\alpha_n}(X_t) \dotsb \tau^{\alpha_1} (X_t) \Omega } \d\alpha_n \dotsb \d \alpha_1.
	$$
	Let $\gamma'\in(0,\gamma)$ and
	\begin{align*}
	&\mathcal{G}_{\gamma'}:=\{(\rho+1)^{1/2}\e^{\alpha h_0} \Exp{\i s h} \f \oplus 0, (\rho+1)^{1/2} \e^{\alpha h_0}  \Exp{\i s h} \ones \oplus 0,\\
			 &\qquad \qquad 0 \oplus \rho^{1/2} \e^{\alpha h_0} \Exp{-\i s \overline{h}} \overline{\f} , 0 \oplus \rho^{1/2} \e^{\alpha h_0} \Exp{-\i s \overline{h}} \ones\}_{s\in\rr, |\Im \alpha|<\frac12\gamma'}.
	\end{align*}
	Lemma \ref{lemma:max-exp-x} implies,
	\begin{equation}\label{eq:uniform_bound_cG_gamma}
	m:=\sup_{\phi\in\mathcal{G}_{\gamma}}\|\phi\|<\infty.
	\end{equation}
	Note that the operators $\tau^{\alpha_i}(X_t)$ are non-commutative polynomials of degree~$2$ in $a^\sharp(\phi),\ \phi\in{\mathcal{G}_{\gamma}}$. Each one of this polynomial consists of 16 monomials.

	Using Wick's Theorem, each integrand in the truncated series $F_{t,N}$ admits and analytic extension to~$\{(\alpha_1,\ldots,\alpha_n)\in\cc^n\ |\ |\Im \alpha_j|<\frac12\gamma\}$. Hence each $F_{t,N}$ admits an analytic extension to~$\{\alpha \in \cc : |{\Im \alpha}|<\frac 12 \gamma \}$.

	Lemma~\ref{lem:reaminder-for-real} yields convergence of the sequence $(F_{t,N})_{N \in \nn}$ on a subset of $\{\alpha \in \cc : |{\Im \alpha}|<\tfrac 12 \gamma \}$ that has an accumulation point and the limit there coincides with the characteristic function of $\P_t$.

	From \eqref{eq:uniform_bound_cG_gamma},
	\begin{align*}
		&\left\|\int_0^\alpha \dotsi \int_0^{\alpha_{n-1}} \braket{\Omega, \tau^{\alpha_n} (X_t) \dotsb \tau^{\alpha_1} (X_t) \Omega} \d\alpha_n \dotsb \d \alpha_1\right\| \\
		&\hspace{30ex} \leq  \frac{|\alpha|^n}{n!} \| \tau^{\alpha_n} (X_t) \dotsb \tau^{\alpha_1} (X_t) \Omega \| \\
		&\hspace{30ex}\leq \frac{\sqrt{(2n+1)!}}{n!}   (16m^2|\alpha|)^n.
	\end{align*}
	Since $\limsup_{n\to\infty}2^{-n}\frac{\sqrt{(2n+1)!}}{n!}<\infty$, the sequence $(F_{t,N})_{N \in \nn}$ is uniformly bounded on ${\mathcal{D}}:=\{\alpha \in \cc : |\alpha|<\frac{1}{32 m^2}\mbox{ and }|\Im \alpha|<\frac12\gamma'\}$. We conclude by the Vitali--Porter convergence theorem that  $(F_{t,N})_{N \in \nn}$ converges on ${\mathcal{D}}$ to an analytic function.
	Therefore, for $\gamma' > 0$ small enough,
	\begin{align*}
		\sup_{t \in \rr} \ee_t[\Exp{\gamma'|\Delta Q|}] &\leq \sup_{t \in \rr}(\braket{\Omega, \Exp{\gamma' (L + X_t)}\Omega} + \braket{\Omega, \Exp{-\gamma' (L + X_t)}\Omega})
	\end{align*}
	is finite.
\end{proof}

 	\subsection{Proofs for the van Hove Hamiltonian}
		\begin{proof}[Proof of Theorem \ref{thm:vanHove_Poisson}]
Using $L\Omega=0$, the characteristic function of $\P_t$ can be rewritten as
\begin{equation}\label{eq:bosons_lin_rewrite_cht}
\begin{split}
		\cht(\alpha)&= \braket{\Exp{-\i t(L + \varphi^{\text{AW}}(f))}\Omega,  \Exp{\i \alpha L} \Exp{-\i t(L + \varphi^{\text{AW}}(f))} \Exp{-\i \alpha L} \Omega} \\
			&= \braket{\Exp{-\i t(L + \varphi^{\text{AW}}(f))}\Omega,  \Exp{-\i t(L + \varphi^{\text{AW}}(\Exp{\i \alpha h_0}f))}  \Omega}.
	\end{split}
\end{equation}
Using the Lie--Trotter--Kato formula, \eqref{eq:weyl} and $L\Omega=0$, for any~$t$ and~$\alpha$ in~$\rr$,
\begin{equation}\label{eq:bosons_lin_TK}
\begin{split}
	\Exp{-\i t(L + \varphi^{\text{AW}}(\Exp{\i \alpha h_0}f))}\Omega &= \lim_{n\to\infty} (\Exp{-\i \frac{t}{n}L} \Exp{-\i\frac{t}{n}\varphi^{\text{AW}}(\Exp{\i \alpha h_0} f)})^n\Omega\\
	&= \Exp{\i\theta} \lim_{n\to\infty}\Exp{-\i \varphi^{\text{AW}}(\sum_{k=1}^n \frac{t}{n} \Exp{\i \alpha h_0} \Exp{-\i \frac{k t}{n} h_0} f) }\Omega,
\end{split}
\end{equation}
where $\Exp{\i\theta}$ is a phase that is independent of~$\alpha$ and thus cancels in~\eqref{eq:bosons_lin_rewrite_cht}.

Let $\mathcal{C}:=\operatorname{linspan}\{a^*(\phi_n)\dotsb a^*(\phi_1)\Omega : n\in\nn, \phi_i\in\Dom(h_0^{-1/2}\oplus h_0^{-1/2})\}.$ This is a core for all of the field operators
$\varphi^{\text{AW}}(\sum_{k=1}^n \frac{t}{n} \Exp{\i \alpha h_0} \Exp{-\i \frac{k t}{n} h_0} f)$
and
$\varphi^{\text{AW}}\left(\i\e^{\i\alpha h_0}h_0^{-1}(\e^{-\i t h_0}-1)f\right)$,
where $x\mapsto x^{-1}(\e^{-\i tx}-1)$ is extended to~$x=0$ by continuity. From the convergence of the Riemann sums to a Riemann integral,
$$
	\lim_{n\to\infty} \varphi^{\text{AW}}\left(\sum_{k=1}^n \frac{t}{n} \Exp{\i \alpha h_0} \Exp{-\i \frac{k t}{n} h_0} f\right)\Psi=\varphi^{\text{AW}}\left(\i\e^{\i\alpha h_0}h_0^{-1}(\e^{-\i t h_0}-1)f\right)\Psi.
$$
for every $\Psi\in\mathcal{C}$. It then follows from Proposition~VIII.25 of~\cite{RS1} that
$$\srlim_{n\to\infty} \varphi^{\text{AW}}\left(\sum_{k=1}^n \frac{t}{n} \Exp{\i \alpha h_0} \Exp{-\i \frac{k t}{n} h_0} f\right)=\varphi^{\text{AW}}\left(\i\e^{\i\alpha h_0}h_0^{-1}(\e^{-\i t h_0}-1)f\right).$$
Hence, from \eqref{eq:bosons_lin_TK},
\begin{align*}
\Exp{-\i t(L + \varphi^{\text{AW}}(\Exp{\i \alpha h_0}f))}\Omega &=\Exp{\i\theta} \Exp{-\i\varphi^{\text{AW}}\left(\i\e^{\i\alpha h_0}h_0^{-1}(\e^{-\i t h_0}-1)f\right)}\Omega.
\end{align*}
Setting $f_t:=-\i h_0^{-1}(\e^{-\i t h_0}-1)f$, from this equality and \eqref{eq:bosons_lin_rewrite_cht}, relations \eqref{eq:weyl}  and  \eqref{eq:weyl-eval}, imply
	\begin{align*}
		\cht(\alpha) &= \braket{
					W^{AW}(f_t) \Omega,
					W^{AW}(\e^{\i\alpha h_0}f_t) \Omega} \\
					&=\braket{\Omega, W^{AW}((\e^{\i\alpha h_0}-1)f_t)\Omega}\e^{\frac{\i}{2}\Im(\braket{f_t,(\e^{\i\alpha h_0}-1)f_t})}\\
				&=\Exp{\frac12\braket{(1+\rho)^{1/2}f_t,(\e^{\i\alpha h_0}-1)(1+\rho)^{1/2}f_t}+\frac12\braket{\rho^{1/2}f_t,(\e^{-\i\alpha h_0}-1)\rho^{1/2}f_t}}.
	\end{align*}
	Writing explicitly the inner products as integrals with respect to~$\d e$, using $1+\rho=(1-\e^{-\beta h_0})^{-1}$ and a change of variable $e\to -e$ in the second inner product,
$$
   \cht(\alpha)=\exp\left(\int_\rr (\e^{\i \alpha e} -1)\d\nu_t(e)\right).
$$
L\'evy--Khintchine's canonical representation of infinitely divisible distributions yields the proposition.
\end{proof}

\begin{proof}[Proof of Theorem \ref{thm:bosons-lin-equiv-moments}]
The implication \ref{it:uniform_bound_moments_bosons_lin}$\implies$\ref{it:integrable_moments_bosons_lin} is obvious.

Since Proposition \ref{thm:vanHove_Poisson} implies $\P_t$ is the measure of an inhomogeneous Poisson process with intensity $\d \nu_t(e)$, its $2n+2^{\text{nd}}$ cumulant is given by $\int_\rr e^{2n+2}\d \nu_t(e)$. Then \ref{it:an_fun_bosons_lin} implies this cumulant is uniformly bounded and the equivalence between the existence of a uniform bound on the $2n+2^{\text{nd}}$ cumulant and $2n+2^{\text{nd}}$ moment yields \ref{it:uniform_bound_moments_bosons_lin}.

For the implication \ref{it:integrable_moments_bosons_lin}$\implies$\ref{it:diff_fun_bosons_lin}, we first show that \ref{it:integrable_moments_bosons_lin} implies the $2n+2^{\text{nd}}$ cumulant of $\P_t$ is integrable over $(t_1,t_2)$. The $2n+2^{\text{nd}}$ cumulant of $\P_t$ is a polynomial in $\{\ee_t(\Delta Q^p)\}_{p=0}^n$ where in each monomial the product of the moments $\prod_{i=1}^k\ee_t(\Delta Q^p_i)$ is such that $\sum_{i=1}^k p_i=2n+2$. Using Jensen's inequality on each moment in the product,
$$\prod_{i=1}^k\ee_t(\Delta Q^{p_i})\leq \ee_t(\Delta Q^{2n+2})^{\frac{1}{2n+2}\sum_{i=1}^kp_i}=\ee_t(\Delta Q^{2n+2}).$$
Hence \ref{it:integrable_moments_bosons_lin} implies the $2n+2^{\text{nd}}$ cumulant of $\P_t$ is integrable over $(t_1,t_2)$.

Since the $2n+2^{\text{nd}}$ cumulant of $\P_t$ is given by $\int_\rr e^{2n+2}\d \nu_t(e)$,
$$\int_{t_1}^{t_2}\int_\rr e^{2n+2}\d \nu_t(e)\d t<\infty.$$
By Fubini's Theorem,
$$\int_\rr \left(\int_{t_1}^{t_2}(1-\cos(et))\d t\right)e^{2n}|\rho(e)||f(|e|)|^2\d e<\infty.$$
Since $e\mapsto\int_{t_1}^{t_2}(1-\cos(et))\d t=(t_2-t_1)(1-\frac{\sin(e t_2)-\sin(e t_1)}{e(t_2-t_1)})$ is lower bounded away from $0$ on $\rr\setminus[-1,1]$,
$$\int_{1}^{\infty} e^{2n}\rho(e)|f(e)|^2\d e<\infty.$$
Since $\rho(e)\geq 1$ for $e\geq 1$,
$$\int_{1}^\infty e^{2n}|f(e)|^2\d e<\infty.$$
Finally $f\in L^2(\rr_+)$ implies $\int_{0}^1e^{2n}|f(e)|^2\d e\leq \int_{0}^1|f(e)|^2\d e$. Hence
$$\int_{\rr_+}e^{2n}|f(e)|^2\d e<\infty.
 \eqno\qedhere
$$
\end{proof}
\begin{proof}[Theorem \ref{thm:bosons-lin-equiv-an}]
The implication \ref{it:unif_bound_fourier_bosons_lin}$\implies$\ref{it:integrable_fourier_bosons_lin} is obvious.

From Proposition \ref{thm:vanHove_Poisson}, the characteristic function of $\P_t$ is,
$$\cht(\alpha)=\exp\left(\int_\rr (\e^{\i e\alpha}-1)\d\nu_t(e)\right).$$
Assumption $f\in \Dom \e^{\frac 12\gamma\hat e}\cap \Dom \hat e^{-1/2}$ implies $f\in\Dom \hat e^{-1/2}\e^{\frac12 \gamma\hat e}$. Hence, since the functions $e\mapsto |1-\cos(et)|/e^2$ and $e\mapsto |e|/|1-\e^{-\beta e}|$ are both bounded in $e\in\rr$, \ref{it:an_fun_bosons_lin} and the expression of $\d\nu_t$ imply $\sup_{t\in\rr}\int_\rr \e^{\gamma |e|}\d\nu_t(e)<\infty$. The inequality $\int_\rr \e^{\gamma |e|}\d\nu_t(e)\geq \int_\rr \e^{\pm\gamma e}\d\nu_t(e)$ and the continuity of $x\mapsto \Exp{x}$ imply
$$\sup_{t\in\rr}\cht(\pm\i\gamma)<\infty.$$
The inequality $\e^{|x|}\leq \e^x+\e^{-x}$ for $x\in\rr$ yields \ref{it:an_fun_bosons_lin} $\implies$ \ref{it:unif_bound_fourier_bosons_lin}.

Using the $\e^{|x|}\geq \frac12(\e^x+\e^{-x})$ and Jensen's inequality \ref{it:integrable_fourier_bosons_lin} implies, $t\mapsto \log \cht(\pm\i\gamma)$ are both integrable over $(t_1,t_2)$. Repeating the same steps as in the proof of Theorem~\ref{thm:bosons-lin-equiv-moments},
$$\int_0^{\infty}\e^{\gamma e}|f(e)|^2\d e<\infty,$$
and we conclude that \ref{it:integrable_fourier_bosons_lin} implies \ref{it:an_fun_bosons_lin}.
\end{proof}
\begin{remark}
From these last two proofs, one can understand why, in the second statements of Theorems~\ref{thm:fermions-equiv-moments}, \ref{thm:bosons-quad-equiv-moments}, \ref{thm:bosons-lin-equiv-moments} and~\ref{thm:bosons-lin-equiv-an}, the integrability over an interval~$(t_1,t_2)$ cannot be replaced by finiteness for a~$t\in\rr$.

Let
$$f_n:e\mapsto\begin{cases}\frac{1}{\lceil e\rceil^{n+1}}(\lceil e\rceil - e -\i \lceil e\rceil^{-1})^{-1}&\mbox{if }e\geq 1\\ \i e&\mbox{if }0\leq e<1.\end{cases}$$
For $n\in\nn\setminus\{0\}$, $f_n\in \Dom \hat e^{-1/2}\cap \hat e$. Integrating between two consecutive integers $N>1$,
$$\int_{N-1}^{N} e^{2n}|f_n(e)|^2\d e\geq \left(\frac{(N-1)}{N}\right)^{2n}\frac{\operatorname{arctan}N}{N}.$$
Since $\lim_{N\to\infty} \operatorname{arctan}N= \tfrac{\pi}{2}$, $\lim_{N\to\infty} \tfrac{N-1}{N}=1$ and $\sum_{N=2}^\infty \tfrac 1N=\infty$, $f_n\not\in \Dom \hat e^{n}$ but
$$\int_{N-1}^{N}e^{2n}(1-\cos(e2\pi))|f_n(e)|^2\d e\leq \frac1{N^2}\int_0^{1}\frac{1-\cos(x2\pi)}{x^2+N^{-2}}\ dx\leq \frac{2\pi^2}{N^2},$$
where the last inequality follows from $\frac{1-\cos(2x)}{x^2}\leq 2$. Then from $\sum_{N=2}^\infty\tfrac{1}{N^2}<\infty$ and Theorem \ref{thm:vanHove_Poisson}, $\ee_{2\pi}(\Delta Q^{2n+2})$ is finite even if $f\not\in \Dom \hat e^n$.

We can similarly construct a $f_\gamma\in \Dom \hat e^{-1/2}\cap \hat e$ such that $\ee_{2\pi}(\e^{\gamma|\Delta Q|})$ is finite but $f_\gamma\not\in \Dom \e^{\frac12\gamma \hat e}$.
\end{remark}

\appendix

\section{Self-adjointness}\label{app:self-adjoint}

The following two technical lemmas ensure mathematical soundness of the construction of the perturbed dynamics~$(\tau_V^t)_{t \in \rr}$ and of the measure~$\P_t$ in our models where the perturbation is unbounded~\cite{DJP}. More detailed proofs are provided in~\cite{Renaud-th}.

\begin{lemma}\label{lemma_selfa_impurity}
Let $L+V$  the operator defined in \eqref{eq:def-L-wwA} for the quantum open harmonic oscillator of Section~\ref{sec:model-bosons}. Assume $f\in \Dom(\hat{e})\cap\Dom(\hat{e}^{-1/2})$.
Then $L+V$ is essentially self-adjoint on $\Dom L \cap \Dom V$.
\end{lemma}

\begin{proof}
We apply  the Nelson commutator theorem in the form of \cite[\S{X.37}]{RS2}.
Set $\hat{N}:=\d \Gamma(h_0\oplus h_0+\one)+\one$. Then it is easy to show $D=\Dom(L)\cap \Dom(V)$ is a core for $\hat{N}$. The estimates
\[
\norme{(L+V)\Psi} \leq c\norme{\hat{N}\Psi}
\]
for some $c>0$ and all $\Psi\in D$, and
\[
|((L+V)\Psi,\hat{N}\Psi)-(\hat{N}\Psi, (L+V)\Psi)|\leq d\norme{\hat{N}^{\frac{1}{2}}\Psi}^2,
\]
for some $d > 0$ and all $\Psi \in D$ follow from the well-known estimate
\[
\Big\|\Big(\prod_{i=1}^{n} a^\# (g_i)\Big) (\dG(\one)+\one)^{-\frac{n}{2}} \Big\| \leq c_n\prod_{i=1}^{n} \norme{g_i} .
\]
for $g_i\in \ch$, provided  $\norme{ \rho^{\frac{1}{2}}\f}$, $\norme{ (1+\rho)^{\frac{1}{2}}\f}$, $\norme{h_0 \rho^{\frac{1}{2}}\f}$,$\norme{ h_0(1+\rho)^{\frac{1}{2}}\f}$ are finite.
The hypothesis  $f\in \Dom(\hat{e})\cap\Dom(\hat{e}^{-1/2})$ precisely guarantees that these norms are finite.
\end{proof}

\begin{lemma}\label{lemma_selfa_linear}
Let $L+\phi^{\textnormal{AW}}(f)$  the operator defined in \eqref{eq:def-L-vH} in the van Hove model of Section~\ref{sec:model-bosons-lin}. Assume $f\in \Dom(\hat{e})\cap\Dom(\hat{e}^{-1/2})$.
Then $L+\phi^{\textnormal{AW}}(f)$ is essentially self-adjoint on $\Dom L \cap \Dom \phi^{\textnormal{AW}}(f)$.
\end{lemma}

\begin{proof} The result follows applying Nelson's commutator theorem in a similar way as in Lemma~\ref{lemma_selfa_impurity}.
\end{proof}
 \newcommand{\p}{\mathsf{p}}
\section{A note on the thermodynamic limit} \label{app:TL}
	\newcommand{\Weyl}{\mathfrak{W}}
In this appendix, we show that for the models of Section~\ref{sec:models}, $\P_t$ is the weak limit of a family of probability measures associated to a two-time measurement protocol on finite dimensional systems as described in Section~\ref{sec:Cstar-desc}. We provide the proof for bosonic models; the fermionic case only requires a straightforward adaptation of the arguments.

\medskip
Consider a positive semi-definite densely defined self-adjoint operator~$h_0$ on a separable complex Hilbert space~$\h$.
The extended system we wish to obtain in the thermodynamic limit is a quasi-free Bose gas with one-particle unperturbed Hamiltonian~$h_0$, at equilibrium at inverse temperature~$\beta > 0$, together with a perturbation of the form $$V = \tfrac{1}{\sqrt{2}}(a^*(g) + a(g))$$ with $g \in \Dom h_0^{-1/2}$ or of the form $$V = a^*(g_1)a(g_2) + a^*(g_2)a(g_1)$$
with $g_1, g_2 \in \Dom h_0^{-1/2}$. These operators are naturally affiliated to the Weyl algebra over~$\Dom h_0^{1/2}$, denoted~$\Weyl(\Dom h_0^{-1/2})$.

We work in a representation of this gas on the Araki--Woods von Neumann algebra $\mathfrak{M}^\textnormal{AW} \subset \mathcal{B}(\Gs(\h \oplus \h))$ generated by the Weyl operators
$$
   W^{\textnormal{AW}}(\phi):= W(\sqrt{\one+\rho(h_0)}\phi\oplus \sqrt{\rho(h_0)}\bar \phi)
$$
for $\phi \in \Dom h_0^{-1/2}$ and with initial state~$\braket{\Omega, {\,\cdot\,} \Omega}$ where~$\Omega \in \Gs(\h \oplus \h)$ is the vacuum vector. The Liouvillean for the unperturbed dynamics in this representation is
$$
   L := \dG(h_0 \oplus -h_0).
$$
In this representation, the creation operators affiliated to~$\Weyl(\Dom h_0^{-1/2})$ take the form
$$
    {a^{\textnormal{AW}}}^*(\phi)=a^*(\sqrt{\one+\rho(h_0)}\phi\oplus 0)+a(0\oplus \sqrt{\rho(h_0)}\bar \phi).
$$
for any $\phi\in\Dom h_0^{-1/2}$.
In this manner,~$V$ corresponds to an operator~$V^{\textnormal{AW}}$ affiliated to~$\mathfrak{M}^{\textnormal{AW}}$.

\medskip
Let~$(\p_D)_{D \in \nn}$ be a sequence of orthogonal projections on~$\h$ with $D$-dimensional range and satisfying the following properties:
\begin{enumerate}[label=(\roman*)]
\item the sequence~$(\p_D)_{D\in\nn}$ converges strongly to the identity; \item for all $D\in\nn$, $\overline{\p_D \phi}=\p_D\bar \phi$ for all $\phi \in \h$;
\item for all $D\in\nn$, $\p_Dh_0\p_D$ is positive definite;
\item there exists a core~$C\subseteq\h$ for $h_0$ such that $\lim_{D \to \infty} \p_Dh_0\p_D\phi=h_0\phi$ for all~$\phi \in C$.\label{it:s_lim_h_0}
\end{enumerate}
\noindent
With the help of this family of projectors, we define, for each~$D \in \nn$, an unperturbed Hamiltonian $H_0^{(D)} := \dG(\p_D h_0 \p_D)$. We also set $H_0^{(D,\delta)} := \dG(p_\delta(\p_D h_0 \p_D))$ for the purpose of defining an initial state below, understood in the sense of functional calculus for the real function
$$
   p_\delta : e \mapsto
      \max\{\delta,e\}.
$$
The parameter~$\delta > 0$ serves as an IR regularization.
Both $H_0^{(D,\delta)}$ and $H_0^{(D)}$ are self-adjoint operators on~$\Gs(\p_D \h)$ and, since $\p_Dh_0\p_D>0$, the operator $\Exp{-\beta H_0^{(D,\delta)}}$ is trace class.

Let $V^{(D)}$ be the same polynomial as~$V$, but where $g$ [resp. $g_1$ and $g_2$] is replaced by $\p_D g$ [resp. $\p_D g_1$ and $\p_D g_2$].

We make the following definition, in accordance with an extension of Remark~\ref{rk:finite_dim} from matrices to trace class operators.
\begin{definition}\label{def:ttm:fd}
	Let $\P_t^{(D,\delta)}$ be the probability measure of the heat variation as defined by the two-time measurement of the unperturbed Hamiltonian~$H_0^{(D)}$ with respect to the perturbation $V^{(D)}$ and the initial state
	\begin{align*}
	   \omega^{(D,\delta)} :
			\Weyl(\p_D \h) &\to \cc \\
			A &\mapsto \frac{\tr(A \Exp{-\beta H_0^{(D,\delta)}})}{\tr(\Exp{-\beta H_0^{(D,\delta)}})}.
	\end{align*}
	Namely, it is the probability measure with characteristic function
	$$\cht^{(D,\delta)}(\alpha) = \omega^{(D,\delta)}(\e^{\i\alpha\e^{\i t(H_0^{(D)}+V^{(D)})}H_0^{(D)}\e^{-\i t(H_0^{(D)}+V^{(D)})}}\e^{-\i\alpha H_0^{(D)}}).$$
\end{definition}

The algebra~$\Weyl(\p_D \h)$ together with the state~
$
   \omega^{(D,\delta)}
$
admit an Araki--Woods representation on $\Gs(\h\oplus\h)$: the vacuum~$\Omega \in \Gs(\h \oplus \h)$ is a vector representative of~$\omega^{(D,\delta)}$ and the $*$-isomorphism
$
     \pi^{\textnormal{AW}(D,\delta)} : \Weyl(\p_D \h) \to \mathfrak{M}^{\textnormal{AW}(D,\delta)}
$
is the extension by linearity of
$
     \pi^{\textnormal{AW}(D,\delta)}(W(\phi)):=W^{\textnormal{AW}(D,\delta)}(\phi)
$
for $\phi \in \p_D \h$, where
$$
   W^{\textnormal{AW}(D,\delta)}(\phi) := W(\sqrt{\one + \rho\circ p_\delta(\p_D h_0 \p_D)}\phi \oplus \sqrt{\rho\circ p_\delta(\p_D h_0 \p_D)}\,\overline\phi )
$$
and
$\rho: \rr_+\ni e\mapsto (\e^{\beta e}-1)^{-1}$.
For each~$D$ and~$\delta$, $\mathfrak{M}^{\textnormal{AW}(D,\delta)} $
is the von Neumann algebra generated by $\{W^{\textnormal{AW}(D,\delta)}(\phi):\phi\in\p_D\h\}$.
It is a subalgebra of the algebra~$\Weyl(\h \oplus \h)$ of operators on~$\Gs(\h \oplus \h)$. The Liouvillean for the dynamics implemented by~$H_0^{(D)}$ is
$$
   L^{(D)} := \dG(\p_D h_0 \p_D \oplus - \p_D h_0 \p_D).
$$

From its definition as the generator of the unitary groups $t\mapsto W(tg)$, in this representation, each creation operator~$a^*(g)$ on $\Gs(\p_D\h)$ is mapped to
$$
   a^*(\sqrt{\one + \rho\circ p_\delta(\p_D h_0 \p_D)}g \oplus 0) + a(0 \oplus \sqrt{\rho\circ p_\delta(\p_D h_0 \p_D)}\,\overline g)
$$
on~$\Gs(\h \oplus \h)$.
In this manner, the perturbation~$V^{(D)}$ on $\Gs(\p_D\h)$ is mapped to an unbounded operator~$V^{\textnormal{AW}(D,\delta)}$ affiliated to~$\mathfrak{M}^{\textnormal{AW}(D,\delta)}$.

\medskip
The following lemma allows us to prove a sufficiently strong notion of convergence of $V^{\textnormal{AW}(D,\delta)}$ towards $V^{\textnormal{AW}}$. It is a straigforward exercise in analysis on Fock space using the inequalities~\eqref{eq:bound_vector_aOmega} and $\|(\dG(\one)+\one)^{-1/2}a^\sharp(g)\|\leq \|g\|$.
\begin{lemma}\label{lem:conv_monome_on_core}
      Let $\Phi \in \operatorname{linspan}\{a^*(g_k)\dotsb a^*(g_1)\Omega: k\in\nn, g_1,\dotsb g_k\in\mathfrak{h}\}$. Then, the function
   $$\mathfrak{h}^n\ni(x_1, \dotsc, x_n) \mapsto a^{\sharp}(x_n)\dotsb a^{\sharp}(x_1)\Phi \in \Gs(\mathfrak{h})$$
   where each $\sharp$ stands for $*$ or nothing, is norm continuous for $\mathfrak{h}^n$ equipped with the product of the norm topology.
\end{lemma}

\begin{proposition}\label{prop:conv-on-core}
   Let $\mathcal C:=\operatorname{linspan}\{a^*(g_n)\dotsb a^*(g_1)\Omega: n\in\nn, g_1,\dotsc,g_n\in C\}$ with $C$ a core for $h_0$ verifying \ref{it:s_lim_h_0}. The space $\mathcal C$ is a common core for $L^{(D)}, L, V^{\textnormal{AW} (D,\delta)}$ and $V$. Moreover, for all $\Phi \in \mathcal C$,
	$$
       \lim_{D \to \infty} L^{(D)} \Phi = L\Phi
   $$
	and
   $$
      \lim_{\delta \downarrow 0} \lim_{D \to \infty} V^{\textnormal{AW} (D,\delta)} \Phi = V^{\textnormal{AW}}\Phi.
   $$
\end{proposition}

\begin{proof}
	The fact that $\mathcal C$ is a common core is direct from the definitions of the operators and the fact that $C$ is a core for $h_0$.

	The first convergence follows from a direct computation in Fock space and \ref{it:s_lim_h_0}.

   For the second limit, let $g \in \Dom h_0^{-1/2}$. By Proposition~VIII.25 of~\cite{RS1} assumption \ref{it:s_lim_h_0} implies the sequence $\p_Dh_0 \p_D$ converges in the strong resolvant sense to $h_0$. Then, for any $\delta > 0$, since $e\mapsto \rho\circ p_\delta(e)$ is continuous and bounded,
   $$
      \lim_{D \to \infty} \sqrt{\rho \circ p_\delta(\p_Dh_0\p_D)}\, \p_D \overline g = \sqrt{\rho \circ p_\delta(h_0)}\, \overline g.
   $$
   By Lebesgue's monotone convergence Theorem,
   $$
      \lim_{\delta \downarrow 0} \sqrt{\rho \circ p_\delta(h_0)}\, \overline g = \sqrt{\rho(h_0)}\, \overline g = \sqrt{\rho} \,\overline g.
   $$
   The same arguments lead to
   $$\lim_{\delta \downarrow 0}\lim_{D \to \infty} \sqrt{1 + \rho \circ p_\delta(\p_Dh_0\p_D)}\, \p_D g=\sqrt{1 + \rho(h_0)}\,g.$$
   The second convergence therefore follows from Lemma \ref{lem:conv_monome_on_core} with $n=1$ or $n=2$.
\end{proof}

For the proof of the self-ajointness of the different operators, we refer the reader to Appendix~\ref{app:self-adjoint}.
\begin{proposition}
	With $\P_t^{(D,\delta)}$ as in Definition~\ref{def:ttm:fd} and~$\P_t$ the spectral measure of
	$$L+V^{\textnormal{AW}}-\e^{\i t (L+V^{\textnormal{AW}})}V^{\textnormal{AW}}\e^{-\i t (L+V^{\textnormal{AW}})}$$
	with respect to $\Omega$,
	$$
		\wlim_{\delta\downarrow 0} \wlim_{D \to \infty} \P_t^{(D,\delta)} = \P_t.
	$$
\end{proposition}

\begin{proof}
	By extension of Remark~\ref{rk:finite_dim} to the trace-class operator $\Exp{-\beta H_0^{(D,\delta)}}$, the characteristic function of~$\P_t^{(D,\delta)}$ can be expressed in the standard GNS representation, which is unitarily equivalent to the Araki--Woods representation. Hence,
	\begin{align*}
		\mathcal{E}_t^{(D,\delta)}(\alpha) = \braket{\Exp{-\i t (L^{(D)} + V^{\textnormal{AW}(D,\delta)})}\Omega,  \Exp{\i\alpha L^{(D)}} \Exp{-\i t (L^{(D)} + V^{\textnormal{AW}(D,\delta)})}  \Omega}
	\end{align*}
	By Proposition~\ref{prop:conv-on-core} above and Proposition~VIII.25 of~\cite{RS1}, $L^{(D)} + V^{\textnormal{AW}(D,\delta)}$ converges in the strong resolvent sense to $L + V^{\textnormal{AW}}$. Hence,
	$$
		\lim_{\delta \downarrow 0} \lim_{D \to \infty} \Exp{-\i t (L^{(D)} + V^{\textnormal{AW}(D,\delta)})}\Omega = \Exp{-\i t (L + V^{\textnormal{AW}})}\Omega.
	$$
	Because $\Exp{\i\alpha L^{(D)}}$ and $\Exp{-\i t(L^{(D)} + V^{\textnormal{AW}(D,\delta)})}$ are unitary for $\alpha \in \rr$ and $t \in \rr$, the strong resolvent convergence of~$L^{(D)}$ and $L^{(D)}+ V^{\textnormal{AW}(D,\delta)}$ also imply through an~$\epsilon/2$-argument that
	$$
		\lim_{\delta \downarrow 0} \lim_{D \to \infty} \Exp{\i\alpha L^{(D)}}\Exp{-\i t (L^{(D)} + V^{\textnormal{AW}(D,\delta)})}\Omega = \Exp{\i\alpha L}\Exp{-\i t (L + V^{\textnormal{AW}})}\Omega.
	$$
	By continuity of the inner product, we have
	$$
		\lim_{\delta \downarrow 0} \lim_{D \to \infty} \mathcal{E}_t^{(D,\delta)}(\alpha) = \mathcal{E}_t(\alpha)
	$$
	for all~$\alpha \in \rr$ and the result follows from L\'evy's continuity theorem.
\end{proof}

\bibliographystyle{amsalpha}
\small
\bibliography{fcs-tails-ref}

\end{document}